\documentclass[12pt]{article}

\usepackage{amsfonts,amssymb,amsmath,amsthm,epsfig,euscript,bbm,amsthm}
\usepackage{algorithm}
\usepackage{algpseudocode}
\usepackage{amsthm}

\usepackage{tikz}
\usetikzlibrary{arrows}
\usepackage {graphicx}
\usepackage{enumerate}
\usepackage[margin=1in]{geometry}
\usepackage[font=small,labelfont=bf]{caption}
\usepackage{subcaption}
\usepackage{caption}
\usepackage[authoryear]{natbib}
\usepackage[titletoc,title]{appendix}
\usepackage[pdftex,colorlinks]{hyperref}
\usepackage{mathtools}
 \usepackage{xcolor}

\usepackage{graphicx}
\usepackage{subcaption}
 
 \definecolor{orange}{RGB}{230,170,120}
  \definecolor{green}{RGB}{120,200,120}

 \hypersetup{
 	colorlinks=true,
 	linkcolor=blue,
 	filecolor=blue,      
 	urlcolor=cyan,
 }
 \hypersetup{colorlinks,linkcolor={blue},citecolor={blue},urlcolor={red}} 
 \usepackage{geometry}                
 \usepackage{multirow}
\usepackage[section]{placeins}
 \usepackage{bbm}
 \usepackage{graphicx}
 \usepackage{amssymb}
 \usepackage{epstopdf}
 \usepackage{tikz} 
 \usepackage{amsmath} 
\usepackage[utf8]{inputenc}
\usepackage{algorithm}

\usepackage{amsmath}
\usetikzlibrary{decorations.pathreplacing,matrix}

\tikzset{ 
table/.style={
  matrix of math nodes,
  row sep=-\pgflinewidth,
  column sep=-\pgflinewidth,
  nodes={rectangle,text width=3em,align=center},
  text depth=1.25ex,
  text height=2.5ex,
  nodes in empty cells,
  left delimiter=[,
  right delimiter={]},
  ampersand replacement=\&
}
}

\usepackage{algorithm,algpseudocode}

\newcounter{algsubstate}
\renewcommand{\thealgsubstate}{\alph{algsubstate}}
\theoremstyle{plain}
 \newtheorem{thm}{Theorem}[section]
 
 \newtheorem{prop}[thm]{Proposition}
 \newtheorem{cor}{Corollary}
 
 \theoremstyle{definition}

 \newtheorem{exmp}{Example}[section]
  \newtheorem{ass}{Assumption}
  
 \theoremstyle{definition}
 \newtheorem{rem}{Remark}
 
 \makeatletter
 \def\BState{\State\hskip-\ALG@thistlm}
 \makeatother
 
 \usepackage{authblk}

\def\spacingset#1{\renewcommand{\baselinestretch}%
{#1}\small\normalsize} \spacingset{1}
 
\definecolor{coquelicot}{rgb}{1.0, 0.22, 0.0}
\usepackage[color=orange!20]{todonotes} \presetkeys{todonotes}{inline}{}

 \newcommand{\blind}{0}
 
 \usepackage{setspace}

\begin{document}
  
\if0\blind
{  
  \title{\bf Identification and Inference for Synthetic Controls with Confounding\footnote{We thank Phillip Heiler, David Hirshberg, Jann Spiess, Jesse Shapiro, Stefan Wager, Kaspar W\"uthrich, Yinchu Zhu for helpful comments.}  }
\author{Guido W. Imbens\footnote{Graduate School of Business and Department of Economics, Stanford University. Email: imbens@stanford.edu} $\quad$ Davide Viviano\footnote{Department of Economics, Harvard University. Email: dviviano@fas.harvard.edu}}
    \date{\today}
  \maketitle
} \fi

\if1\blind
{
  \bigskip
  \bigskip
  \bigskip
  \ 
 \\
 
 \ 
 \\

  \medskip
} \fi

\bigskip  	 

\begin{abstract}
\noindent  

This paper studies inference on treatment effects in panel data settings with unobserved confounding.
 We model outcome variables through a factor model with \textit{random} factors and loadings. Such factors and loadings may act as unobserved confounders: when the treatment is implemented depends on time-varying factors, and who receives the treatment depends on unit-level confounders. 
We study the identification of treatment effects and illustrate the presence of a trade-off between time and unit-level confounding. We provide asymptotic results for inference for several Synthetic Control estimators and show that different sources of randomness should be considered for inference, depending on the nature of confounding. We conclude with a comparison of Synthetic Control estimators with alternatives for factor models.

\end{abstract}

\noindent%
{\it Keywords:} Synthetic Control, Panel Data, Difference-in-Differences, Causal Inference. \\
{\it JEL Codes:} C33 
\vfill

\newpage 

 \onehalfspacing

\section{Introduction}

This paper studies identification and inference in panel or longitudinal data settings where a (possibly small) number of units is exposed to treatment from one period $T_0$ onwards. We will refer broadly to these scenarios as Synthetic Control settings. The existing literature has proposed a variety of methods for estimating counterfactual outcomes: controlling for the lagged outcomes (known as horizontal regression), controlling for the control units outcomes \citep[vertical/Synthetic Control regression,][]{abadie2003economic, abadie2010synthetic}, or using information both across time and units as for the Synthetic DiD \citep{arkhangelsky2021synthetic} and matrix completion methods and factor models \citep{athey2021matrix, gobillon2016regional, xu2017generalized}. Most of this literature has focused on settings where the treatment assignment process is not random, and unobservable characteristics of treated and control units can be matched almost surely.\footnote{See for example, the discussion in \cite{ferman2021properties}. } This approach allows the researcher to interpret the task of counterfactual estimation as a prediction problem, while inference must only consider variation in idiosyncratic shocks. This paper studies the properties of synthetic control-type estimators in the presence of unobservable (random) confounders. These confounders cannot be matched exactly. We derive identification conditions that formalize (i) which regression strategy is best suited in the presence of different confounding mechanisms and (ii) which sources of randomness should be considered for inference.

We study settings where potential outcomes follow an interactive factor model, with possibly high dimensional factors and loadings and additive (nonrandom) treatment effects. Unlike previous literature on synthetic control methods, here, \textit{both} the factors and loadings are random variables and can act as confounders. The loadings act as unobserved confounders for \textit{whom} receives the treatment, and the factors act as confounders for \textit{when} the treatment is implemented. For instance, consider the problem of studying the effect of state-level regulation \citep{abadie2010synthetic}. When the regulation is implemented may depend on aggregate factors of the economy, and which state implements the regulation may depend on state-specific (unobserved) characteristics. 
Because of confounding, conditional on the assignment mechanism, the distribution of the factors and loadings can be arbitrary.

We provide identification restrictions corresponding to the vertical, horizontal, and Synthetic DiD regressions within this confounding model. Vertical regression (i.e., controlling for a weighted combination of control units’ outcomes) allows for arbitrary time-level confounders, and it imposes restrictions on unit confounders. It assumes that conditional on who receives the treatment, unit-level confounders over the treated and control units match \textit{in expectations} after appropriately reweighting (for weights summing to one). The leading case is the following: we can match all unit confounders' conditional expectations that are endogenous. The remaining loadings are exogenous with respect to the assignment mechanism. We refer to the latter identification restriction as ``no high dimensional unit confounders" since these scenarios are attained only when a few loadings are endogenous. Vice-versa, horizontal regression (i.e., controlling for a weighted combination of lagged outcomes) allows for arbitrary endogenous unit-level confounders and assumes ``no high dimensional time confounders", inverting the role of the factors and loadings. This finding illustrates the trade-off between the confounding restriction and the regression strategy.

 We then turn to identification strategies robust to either high dimensional units or time confounders. We show that the Synthetic Difference-in-Differences in \cite{arkhangelsky2021synthetic} is robust to either specification: its bias depends on the \textit{product} of two differences; the difference of the time confounders' weighted expectations between treatment and control periods, and the difference of the unit level confounders. Therefore, its bias is zero if there are \textit{either} low dimensional unit confounders or time confounders. This characterization of the bias formalizes double-robustness in Synthetic Control settings with confounding, and, while building on the intuition in \cite{arkhangelsky2021synthetic}, is novel to the literature.

Taking these results as stock, we draw their implications for inference. Inference with time confounders for Synthetic Controls must consider the randomness generated across units by the (exogenous) high dimensional loadings but can condition on time variation induced by the (endogenous) factors. The reason is that the estimator is unbiased only unconditional on the loadings.  For the horizontal regression, instead, we should account for the randomness generated by the factors but not the loadings. Finally, inference with Synthetic DiD must account for randomness over both unit and time dimensions because robust to confounding occurring \textit{either} across units or time, motivating the construction of standard errors that capture both sources of randomness.

We derive asymptotic properties of the estimators and standard errors corresponding to confounding over time, units, or either of the two. We assume that the post-treatment period and the number of treated units grow with the sample but are small relative to the number of control periods and units. (We return to the case of one or few treated units in Appendix \ref{sec:placebo}.) We characterize the convergence rates of each estimator as a function of the treated units $N_1$ and treatment periods $T_1$. The convergence rate of Synthetic Control is between $1/\sqrt{N_1}$, and $1/\sqrt{N_1 T_1}$, with $1/\sqrt{N_1}$ if the factors are arbitrarily endogenous, and $1/\sqrt{N_1 T_1}$ if the factors are exogenous. Surprisingly, the rate of convergence of Synthetic DiD can be \textit{faster} than the one of Synthetic Control with confounding. The key insight is that the Synthetic control's convergence rate is of order $1/\sqrt{N_1 T_1}$ only if both unit and time confounders concentrate around zero. However, for Synthetic DiD, a convergence rate of order $1/\sqrt{N_1 T_1}$ only requires that the time confounders before and after the treatment concentrate around the \textit{same} expectation after reweighting.

We conclude with a discussion on Synthetic Control and factor models. 
Under distributional restrictions on the factors and loadings, for a Synthetic Control regression, convergence rates depend on the rank of a matrix of (low) dimensional unit confounders, as opposed to the rank of the (possibly high dimensional) matrix of factors and loadings. 
However, sufficiently fast convergence rates of the estimated weights also require some restrictions on the distribution of factors before the treatment time occurs since, otherwise we would not be able to estimate consistently such weights. 



Our paper relates to an extensive literature on panel data models and Synthetic Controls. We build on the literature on synthetic controls \citep{abadie2003economic, abadie2010synthetic, doudchenko2016balancing, abadie2022synthetic}, matrix completion methods and synthetic difference-in-differences \citep{abadie2010synthetic, athey2021matrix, arkhangelsky2021synthetic, liu2022practical, arkhangelsky2023large}. 
We contribute to this literature by studying confounding through the unit and time-level confounders. Different from existing analysis with factor models \citep[e.g.][]{athey2021matrix, arkhangelsky2021synthetic}, here, both factors and loadings are random instead of deterministic and low rank, motivating different identification properties and inference. \cite{ben2021augmented} provide useful finite sample upper bound with fixed loadings and factors using a balancing procedure. Here, we clarify conditions for different estimation strategies and variance estimators in the presence of random confounders.

We relate to studies on confounding with panel data, including
\cite{ferman2016revisiting}, \cite{hahn2017synthetic}, \cite{kellogg2021combining}, and \cite{agarwal2022synthetic} for recent works on identification with dynamic treatments. \cite{shi2021theory}, \cite{imbens2021controlling} study instead proximal methods with panel data. Here, we allow confounding over either or both time and units, whereas these references implicitly condition on either (or both) the factors and loadings. Our focus on inference and its connection to confounding is a further distinction from these references.  \cite{shen2022tale} show that horizontal and vertical regression provide the same point estimates in the absence of an intercept and constraint on the weights. Here, we show that these estimators would be biased without an intercept and weights summing to one. Also, the authors consider random assignments while here we motivate the choice of regression strategies and confidence intervals based on the confounding mechanisms. 

We complement the literature on inference with synthetic control and two-way fixed effects, including \cite{bottmer2021design},\cite{chernozhukov2019practical}, \cite{cattaneo2021prediction}, \cite{chernozhukov2021exact}, \cite{viviano2023synthetic}, \cite{imai2021use}. Different from the references above, here we explicitly model the confounding mechanism in the construction of confidence intervals. \cite{shaikh2021randomization} allow for random treatment timing but does not consider unobserved confounders. Finally, we more broadly relate to a larger strand of literature on factor models \cite{bai2009panel}, \cite{moon2017dynamic}, \cite{bai2019rank} among others, which, however, does not directly tackle the problem of confounding for inference on causal effects. We defer an extensive discussion of theoretical comparisons of synthetic control weights and estimators for factor models to Section \ref{sec:factors}.

\section{Setup}

We consider a panel with units and periods 
$$
i \in \{1, \cdots, N_0, \cdots, N_0 + N_1\}, \quad t \in \{1, \cdots, T_0, \cdots, T_0 + T_1\}, 
$$  
respectively. Define $W_{i,t} \in \{0,1\}$ the treatment assignment for unit $i$ at time $t$, with 
$$
W_{i,t} = 1\{i \ge N_0\} 1\{t \ge T_0\}.
$$
 We let both $T_0$ and $N_0$ be random variables. Denote $\Big(Y_{i,t}(0), Y_{i,t}(1)\Big)$ the potential outcomes under treatment and control of unit $i$ at time $t$.  Throughout our discussion, we assume constant and homogeneous treatment effects of the form
 \begin{equation} \label{eqn:1} 
 Y_{i,t}(1) = \tau + Y_{i,t}(0), 
 \end{equation}  
 where $\tau$ denotes the average treatment effect. We return to heterogeneous treatment effects in Remark \ref{rem:heterogeneous}. Researchers observe a matrix of potential outcomes under control depicted below, and potential outcomes $\mathbf{Y}_{i,t}(1)$ only for units $i \ge N_0, t \ge T_0$.

$$ 
\small 
\begin{aligned} 
\mathbf{Y}(0) = \quad \quad 
\begin{tikzpicture}[baseline,decoration=brace]
\matrix (m) [table] {
    Y_{1,1}  \&  Y_{1,2} \& \cdots \& Y_{1, T_0 - 1} \& \& Y_{1, T_0}  \& \cdots \& Y_{1, T_0 +T} \\
    \vdots \& \vdots \& \ddots \& \cdots \& \& \vdots \& \vdots \& \vdots \\
    \vdots \& \vdots \& \ddots \& \cdots \& \& Y_{N_0 - 1, T_0} \& \vdots \& Y_{N_0 - 1, T_0} \\
   Y_{N_0,1}  \&  Y_{N_0,2} \& \cdots \& Y_{N_0, T_0 - 1} \& \& \textbf{?}  \& \cdots \& \textbf{?} \\ 
    \vdots \& \vdots \& \ddots \& \cdots \& \& \vdots \& \vdots \& \vdots \\
     Y_{N_0 + N_1,1}  \&  Y_{N_0 + N_1,2} \& \cdots \& Y_{N_0 + N_1, T_0 - 1} \& \& \textbf{?}  \& \cdots \& \textbf{?} \\
};
    \draw[decorate,transform canvas={xshift=-1.4em},thick] (m-3-1.south west) -- node[left=2pt] {$N_0$} (m-1-1.north west);
    \draw[decorate,transform canvas={yshift=0.5em},thick] (m-1-1.north west) -- node[above=2pt] {$T_0 - 1$} (m-1-4.north east);
\end{tikzpicture}.
\end{aligned}
$$
Our goal is to estimate the average treatment effect $\tau$, which entails estimating the counterfactual outcomes. The literature has proposed three approaches for this problem: \textit{horizontal regression}, \textit{vertical regression} (i.e., Synthetic Control methods, \cite{abadie2010synthetic}), and \textit{matrix completion methods} \citep{athey2021matrix}. The horizontal regression uses \textit{past outcomes} to predict future outcomes of treated units. The vertical regression uses the control units' outcomes to predict the treated units' outcomes in periods $t > T_0$. The matrix completion method leverages the assumption that $\mathbf{Y}(0)$ can be approximated by a \textit{low-rank} factor model  \citep{bai2003inferential}. In this paper, we study how different sources of confounding justify different regression strategies and standard errors.

\subsection{Data generating process and confounding}

We consider the following factor model.

\begin{ass}[High-dimensional factor model] \label{ass:outcomea} Assume that for all $(i,t)$, 
$$
\small 
\begin{aligned} 
Y_{i,t}(0) & =  \Big(\lambda_t + \tilde{\Lambda}_t\Big)^\top \Big(\gamma_i + \tilde{\Gamma}_i\Big) + \iota_{0,i} + \iota_{1, t} + \varepsilon_{i,t}, \\ ||\mathbb{E}[\tilde{\Gamma}_i]||_{\infty} &  = ||\mathbb{E}[\tilde{\Lambda}_t ]||_{\infty} = \mathbb{E}[\varepsilon_{i,t}] = 0, 
\end{aligned}
$$  
for \textit{random variables} $\tilde{\Gamma}_i, \tilde{\Lambda}_t \in \mathbb{R}^{r}, \varepsilon_{i,t} \in \mathbb{R}$, and \textit{constants} $\lambda_t, \gamma_i \in \mathbb{R}^r, \iota_{0,i}, \iota_{0,t} \in \mathbb{R}$, and $r$ possibly unknown to researchers. Let $\varepsilon_{i,t} | N_0, T_0 \sim_{i.i.d.} \mathcal{P}$ with $\mathbb{E}[\varepsilon_{i,t}^3] < \infty, \mathbb{E}[\varepsilon_{i, t}^2] = \sigma_\varepsilon^2 > 0$. 
\end{ass}

Assumption \ref{ass:outcomea} postulates that outcomes follow a factor model. The factor model depends on factors $\tilde{\Lambda}_t, \lambda_t$, and loadings $\tilde{\Gamma}_i, \gamma_i$. 
Here, $\varepsilon_{i,t}$ is an exogenous shock. The main difference between the model in Assumption \ref{ass:outcomea} with methods discussed in previous literature \citep[e.g.][]{athey2021matrix, abadie2010synthetic, arkhangelsky2021synthetic, shen2022tale} is that that we treat the factors and loadings as random variables, possibly acting as confounders. 

Confounding over time and across units follows two independent processes.  
\begin{ass}[Independent processes] \label{ass:inda} $\Big[\Big(\tilde{\Gamma}_i\Big)_{i\ge 1}, N_0\Big] \perp \Big[\Big(\tilde{\Lambda}_t\Big)_{t\ge 1}, T_0\Big]$, and $\Big[(\tilde{\Gamma}_i)_{i \ge 1}, (\tilde{\Lambda}_t)_{t \ge 1}\Big] \perp (\varepsilon_{i,t})_{i \ge 1, t \ge 1}$. 
\end{ass}

Assumption \ref{ass:inda} states that \textit{who} receives the treatment depends on individual confounders and is independent of \textit{when} the treatment is implemented. Assumption \ref{ass:inda} simplifies our theoretical analysis and allows us to study independently two sources of confounding over time and across units.  

\begin{exmp}[Effects of minimum wage on wages] \label{exmp:main} For state-level regulations, such as studying the effect of minimum wage, Assumption \ref{ass:inda} states that which states may implement the policy depends on state-level characteristics, whereas \textit{when} the policy is implemented depends on aggregate (time-varying) factors of the economy.  \qed  
\end{exmp}

\begin{ass} \label{ass:iid} For all $i, t$, $||\tilde{\Gamma}_i||_2, ||\gamma_i||_2, ||\tilde{\Lambda}_t||_2 , ||\lambda_t||_2$ are uniformly bounded almost surely. Let $\tilde{\Gamma}_i \sim_{i.i.d.} \mathcal{P}_\Gamma, \tilde{\Lambda}_t \sim_{i.i.d} \mathcal{P}_{\Lambda}$, with $\mathrm{Var}(\tilde{\Lambda}_t) = \Sigma_{\tilde{\Lambda}}, \mathrm{Var}(\tilde{\Gamma}_i) = \Sigma_{\tilde{\Gamma}}$ being positive definite. 
\end{ass} 
Assumption \ref{ass:iid} states that (i) factors and loadings have bounded $l_2$-norm; (ii) factors and loadings are $i.i.d.$ \textit{unconditional} on the treatment assignment mechanism. Assumption \ref{ass:iid} does not impose restrictions on how factors and loadings relate to the assignment mechanism. Therefore, it does not rule out dependence between factors, loadings, and the assignment mechanism. The bounded $l_2$-norm restriction allows for a growing number of possibly \textit{weak} factors and allows us to control the variance of potential outcomes.\footnote{The assumption that factors are uniformly bounded can be relaxed by subgaussianity at the expense of additional notation.} Finally, note that time trends or seasonality components can be directly incorporated in a time fixed effect, assuming the separability of non-stationary components.






\subsection{Research questions}

Given Equation \eqref{eqn:1}, it follows that 
 \begin{equation} \label{eqn:1a} 
\small 
\begin{aligned} 
(i) & \quad \quad \mathbb{E}\Big[Y_{N,T} \Big| N_0 = N, T_0 = T, \tilde{\Gamma}_N, \tilde{\Lambda}_T\Big] - \mathbb{E}\Big[Y_{N,T}(0) \Big| N_0 = N, T_0 = T, \tilde{\Gamma}_N, \tilde{\Lambda}_T\Big] & = \tau.\\
(ii) & \quad \quad \mathbb{E}\Big[Y_{N,T} \Big| N_0 = N, T_0 = T, \tilde{\Gamma}_N\Big] - \mathbb{E}\Big[Y_{N,T}(0) | N_0 = N, T_0 = T, \tilde{\Gamma}_N\Big] & = \tau \\  (iii) & \quad \quad \mathbb{E}\Big[Y_{N,T}  \Big|  N_0 = N, T_0  = T, \tilde{\Lambda}_T\Big] - \mathbb{E}\Big[Y_{N,T}(0) | N_0 = N, T_0 = T, \tilde{\Lambda}_T\Big] & = \tau.
\end{aligned} 
\end{equation} 

Namely, researchers may estimate counterfactual outcomes conditional on different information sets to recover the same estimand $\tau$. This raises the following questions: 
\begin{itemize}
\item[Q1] \textit{Which conditions on factors and loadings guarantee the identification of $\tau$, and how do these conditions relate to existing regression strategies? }
Since the factors and loadings are random, identification requires conditions on how they relate to treatment assignments.  In Section \ref{sec:identification}, we argue that different regression strategies are motivated by different identification strategies.
\item[Q2] \textit{Which source of randomness should we consider for inference?} Equation \eqref{eqn:1a} shows that we can take differences of expectations conditional on \textit{different} information sets and recover the \textit{same} estimand $\tau$. This creates confusion for inference on $\tau$: estimators of different conditional expectations will present different variances, not only because such estimators are different but also because we can condition on different information sets. In Section \ref{sec:inference_main}, we relate confidence intervals (and sources of randomness) to the underlying identification assumptions.
\item[Q3] \textit{Synthetic Control or factor models?} Equation \eqref{eqn:1a} $(i)$ shows that it suffices to estimate the factors and loadings to recover $\tau$. However, standard estimators for factor models require a low-rank representation. In Section \ref{sec:estimation}, we provide a discussion and comparison with Synthetic Controls in settings with high dimensional (weak) factors. 
\end{itemize}


\section{Identification} \label{sec:identification}

This section studies identification conditional on $(N_0, T_0) = (N, T)$ for given $(N,T)$.

\subsection{Vertical and horizontal identification}

We first discuss identification for the horizontal regression.

\begin{ass}[Limited confoundedness over time] \label{ass:unc1a} Suppose that 
\begin{itemize} 
\item[(A)] For all $i, t$,  for fixed $\gamma_i, \lambda_t$,
$$
Y_{i,t}(0) \perp  \Big[N_0, T_0\Big] \Big| \tilde{\Gamma}_i. 
$$ 
\item[(B)] There exist weights $w_{h, 1}, \cdots, w_{h, T}, v_{h,T+1}, \cdots, v_{h, T + T_1}$ such that for $T_0 = T$
$$
\Big|\Big|\sum_{s< T} w_{h, s} \lambda_s - \sum_{t \ge T} v_{h, t} \lambda_t\Big|\Big|_2 = 0, \quad \sum_{s < T} w_{h, s} = 1, \quad \sum_{t \ge T} v_{h,t} = 1. 
$$ 
\end{itemize} 
\end{ass} 

Assumption \ref{ass:unc1a} states that we can match the endogenous factors ($\lambda_t$) exactly for some pre and post-treatment weights $w_h, v_h$. The remaining (random) factors $\tilde{\Lambda}_t$ are exogenous. Because we require exact matching for each entry of $\lambda_t$, we interpret the second restriction as imposing a sparsity restriction on $\lambda_t$.\footnote{We can relax the second condition in Assumption \ref{ass:unc1a} up to a small error of order $o(N_1^{-1/2} T_1^{-1/2})$.} As noted in Table \ref{tab:1} and discussed more extensively in Section \ref{sec:factors}, Assumption \ref{ass:unc1a} differs from usual low-rank restrictions in factor models, that instead typically assume that we can consistently estimate all interactive fixed effects. Instead, Assumption \ref{ass:unc1a} states that we can find a set of weights such that we can match the (low-dimensional) endogenous factors $\lambda_t$, whereas the remaining ones are exogenous. 


\begin{prop}[Horizontal identification] \label{lem:3a} Let Assumptions \ref{ass:outcomea}, \ref{ass:inda}, \ref{ass:unc1a} hold. Then for all $i \in \{1, \cdots, N, \cdots, N + N_1\}$ 
\begin{equation} \label{eqn:ajh} 
\small 
\begin{aligned} 
&  \sum_{t \ge T} v_{h, t} \mathbb{E}\Big[Y_{i,t}(0) | N_0 = N, T_0 = T, \tilde{\Gamma}_i\Big] = \sum_{s < T} w_{h, s} \mathbb{E}\Big[Y_{i,s} \Big| N_0 = N, T_0 = T, \tilde{\Gamma}_i\Big] + \beta_{0, h}(w_h, v_h)
\end{aligned} 
\end{equation}
for some constant $\beta_{0, h}(w_h, v_h)$, which depends on $T$, $w_h, v_h$ and any weights $w_{h}, v_h$ satisfying Assumption \ref{ass:unc1a}.
\end{prop}   

\begin{proof}[Proof of Proposition \ref{lem:3a}] See Appendix \ref{proof:lem:3a}
\end{proof} 

For known weights $w_h, v_h$, 
Proposition \ref{lem:3a} suggests regressing the outcomes of the control units onto a weighted average of the outcomes in the previous period as discussed in Appendix \ref{sec:estimation}. A similar result holds also for a Synthetic Control (vertical) regression, after inverting the role of the factor and loadings. 

\begin{ass}[Limited confoundedness over units] \label{ass:unc2a} Suppose that 
\begin{itemize} 
\item[(A)] For all $i, t$, for fixed $\gamma_i, \lambda_t$,
$$
Y_{i,t}(0) \perp \Big[N_0, T_0\Big] \Big| \tilde{\Lambda}_t. 
$$  
\item[(B)] There exist weights $w_{v, 1}, \cdots, w_{v, N-1}, v_{v, N}, \cdots, v_{v, N+1}$ such that for $N_0 = N$,
$$
\Big|\Big|\sum_{j < N} w_{v, j} \gamma_j - \sum_{n \ge N} v_{v, n} \gamma_n \Big|\Big|_2 = 0, \quad \sum_{j < N} w_{v, j} = 1, \quad \sum_{n \ge N} v_{v, n} = 1.  
$$ 
\end{itemize} 
\end{ass}


\begin{prop}[Vertical identification] \label{lem:verticala}  Let Assumptions \ref{ass:outcomea}, \ref{ass:inda}, \ref{ass:unc2a} hold. Then for all $t \in \{1, \cdots, T, \cdots, T + T_1\}$, 
$$
\small 
\begin{aligned} 
 \sum_{n \ge N} v_{v, n} \mathbb{E}\Big[Y_{n, t}(0) | N_0 = N, T_0 = T, \tilde{\Lambda}_t\Big] =  \sum_{j < N} w_{v, j} \mathbb{E}\Big[Y_{j,t} | N_0 = N, T_0 = T, \tilde{\Lambda}_t\Big] + \beta_{0,v}(w_v, v_v). 
\end{aligned} 
$$
for some constant $\beta_{0,v}(w_v, v_v)$ which only depends on $N$, $w_v, v_v$, and any weights $w_v, v_v$ satisfying Assumption \ref{ass:unc2a}. 
\end{prop}

The proof mimics the one of Proposition \ref{lem:3a} once we invert the loadings and factors. 

Proposition \ref{lem:3a} and \ref{lem:verticala} provide identification restrictions for horizontal and vertical regression. The former assumes no high-dimensional time confounders, and the latter assumes no high-dimensional unit confounders.\footnote{Note that Proposition \ref{lem:3a}, \ref{lem:verticala} hold under (weaker) moment restrictions which only require mean exogeneity of $\tilde{\Lambda}_t, \tilde{\Gamma}_i$ instead of complete exogeneity of such random variables. In that case we can interpret, $\lambda_t = \mathbb{E}[\tilde{\Lambda}_t | T_0 = T], \gamma_i = \mathbb{E}[\tilde{\Gamma}_i | N_0 = N]$, as the conditional expectations of the factors and loadings, conditional on the treatment assignment.} 

In the context of Example \ref{exmp:main}, if we interpret $\tilde{\Gamma}_i$ as sectors in the economy, Assumption \ref{ass:unc2a} states that the decision to introduce the minimum wage may only depend on a few (aggregate) sectors but not on each of the sectors separately.

 \begin{rem}[Intercept and weights summing to one] The presence of an intercept in both vertical and horizontal identification strategies and the weights summing to one guarantee unbiasedness in the presence of time and unit fixed effects. Without such conditions, we would be unable to guarantee unbiasedness in the presence of fixed effects. This differs from \cite{shen2022tale}, who show that vertical and horizontal regressions lead to the same point estimates assuming \textit{no} intercepts and unconstrained weights. \qed 
 \end{rem}


\subsection{Double-robust identification}

Next, we study settings where \textit{either} unconfoundedness over time or units may occur. Consider the population equivalent of the Synthetic Differences-in-Difference (sDiD) in \cite{arkhangelsky2021synthetic}, here augmented with post-treatment weights $v_{v}, v_h$: 
$$
\small 
\begin{aligned} 
\bar{\tau}^{dr}(w_h, w_v, v_h,v_v) & = \sum_{j \ge N, t \ge T} v_{v, j} v_{h, t} \mathbb{E}\Big[\bar{\tau}_{j, t}^{dr}(w_h, w_v) \Big|N_0 = N, T_0 = T\Big] 
\end{aligned}
$$
where 
$$
\small 
\begin{aligned} 
\bar{\tau}_{N,T}^{dr}(w_h, w_v, v_h, v_v) & =  Y_{N,T} - \Big\{\sum_{s < T} w_{h, s} \mathbb{E}\Big[Y_{N,s} \Big| N_0 = N, T_0 = T\Big] + \sum_{j < N} w_{v, j} \mathbb{E}\Big[Y_{j,T} \Big| N_0 = N, T_0 = T\Big] \\ &\quad \quad \quad \quad \quad \quad - \sum_{s<T} \sum_{j < N} w_{h,s} w_{v,j} \mathbb{E}\Big[Y_{j,s} \Big| N_0 = N, T_0 = T\Big]\Big\}.
\end{aligned} 
$$ 

In the following proposition, we formalize double-robustness of $\bar{\tau}^{dr}(\cdot)$.



\begin{prop} \label{cor:2}  Suppose that Assumptions \ref{ass:outcomea}, \ref{ass:inda} hold. Then 
\begin{equation} \label{eqn:dr_main}
\small 
\begin{aligned} 
\tau - & \bar{\tau}^{dr}(w_h, w_v, v_h, v_v) = \Big(\bar{\Gamma}_{pre}(w_v) - \bar{\Gamma}_{post}(v_v) \Big)^\top\Big(\bar{\Lambda}_{pre}(w_h) - \bar{\Lambda}_{post}(v_h)\Big), 
\end{aligned} 
\end{equation} 
where 
$$
\small 
\begin{aligned} 
\bar{\Gamma}_{pre}(w_v) = \sum_{i < N} w_{v, i}  \Big(\mathbb{E}[\tilde{\Gamma}_i | N_0 = N] + \gamma_i\Big), \quad \bar{\Gamma}_{post}(v_v) = \sum_{i \ge N} v_{v, i}  \Big(\mathbb{E}[\tilde{\Gamma}_i | N_0 = N] + \gamma_i\Big) \\ 
\bar{\Lambda}_{pre}(w_h) = \sum_{t < T} w_{h,t} \Big( \mathbb{E}[\tilde{\Lambda}_t | T_0 = T] + \lambda_t\Big), \quad  \bar{\Lambda}_{post}(v_h) = \sum_{t \ge T} v_{h,t}  \Big( \mathbb{E}[\tilde{\Lambda}_t | T_0 = T] + \lambda_t\Big). 
\end{aligned} 
$$
\end{prop}

\begin{proof} See Appendix \ref{proof:dr}. 
\end{proof} 

Proposition \ref{cor:2} shows the robustness properties of the Synthetic DiD method: the method is robust to imperfect match over time \textit{and} units. Here,  $\bar{\Gamma}_{pre}(w_v),  \bar{\Gamma}_{post}(v_v)$ denote the pre and post treatment \textit{expectation} of the loadings, and $\bar{\Lambda}_{pre}(w_h)$, $\bar{\Lambda}_{post}(v_h)$ of the factors (\textit{conditional} on the event $N_0 = N, T_0 = T$). Under Assumption \ref{ass:unc1a}, $\Big(\bar{\Lambda}_{pre}(w_h) - \bar{\Lambda}_{post}(v_h)\Big) = 0$, and under Assumption \ref{ass:unc2a}, $\Big(\bar{\Gamma}_{pre}(w_v) - \bar{\Gamma}_{post}(v_v) \Big) = 0$.

Equation \eqref{eqn:dr_main} connects to previous results in the double robust literature \citep[e.g.][]{farrell2015robust}. However, here we intend double-robustness at the population level for $\bar{\tau}^{dr}$, as a function of the time and unit level mismatch, instead of a function of the estimators' convergence rates. It also differs from the analysis in \cite{arkhangelsky2021synthetic} because we characterize double robustness as a function of the distributional properties of the factors and loadings.

\subsection{Implications for estimators} \label{sec:implications}

We conclude this discussion with a short summary of our findings in this section. 
Assuming that we know the weights and intercepts $w_h, v_h, \beta_h, w_v, v_v, \beta_v$ for the moment, we can construct three estimators 
\begin{equation} \label{eqn:estimators1}
\small 
\begin{aligned} 
\hat{\tau}_n^h(w_h, v_h, \beta_h) & = \sum_{t \ge T} v_{h,t} Y_{n,t}  - \sum_{s < T} w_{n,s} Y_{n,s} + \beta_h \\ 
 \hat{\tau}_t^v(w_v, v_v, \beta_v) & =  \sum_{n \ge N} v_{v,n} Y_{t,n} - \sum_{i < N} w_{v,i} Y_{t,i} + \beta_v \\ 
 \hat{\tau}^{dr}(w_h, w_v. v_h, v_v) & = \sum_{t \ge T, n \ge N} v_{h,t} v_{v, n} \Big\{Y_{t, n}  - \sum_{s < T} w_{h, s} Y_{n,s} - \sum_{i < N} w_{v, i} Y_{i, t} + \sum_{i < N} \sum_{s < T} w_{v, i} w_{h, s} Y_{i,s} \Big\}
\end{aligned} 
\end{equation} 
corresponding to the horizontal regression for unit $n$, vertical for unit $t$, and double robust estimator. For the horizontal and vertical regression, we take
\begin{equation} \label{eqn:estimators2} 
\small 
\begin{aligned} 
\hat{\tau}^h(w_h, v_h, \beta_h, q_h) & = \sum_{n \ge N} q_{h,n} \hat{\tau}_n^h(w_h, v_h, \beta_h), \quad &\sum_{n \ge N} q_{h,n} = 1 \\ 
\hat{\tau}^v(w_v, v_v, \beta_v, q_v)  & = \sum_{t \ge T} q_{v,t} \hat{\tau}_t^v(w_v, v_v, \beta_v), \quad &\sum_{t \ge T} q_{v,t} = 1, 
\end{aligned} 
\end{equation}  
where $q_h, q_v$ are arbitrary weights that sum to one. A simple example is $q_{h, n} = 1/N_1, q_{v,t} = 1/T_1$. The additional weights $q_h, q_v$ and the post-treatment weights $v_h, v_v$  are motivated treatment effect homogeneity, see Remark \ref{rem:heterogeneous}. In the following proposition, we illustrate when each of these estimators is unbiased.

\begin{prop}[Conditioning sets] \label{lem:expectations} Under Assumptions \ref{ass:outcomea}, \ref{ass:inda}, for $n \ge N, t \ge T$, 
$$
\begin{aligned} 
(A) \quad  & \tau - \mathbb{E}\left[\hat{\tau}_n^h(w_h, v_h, \beta_h) \Big| N_0 = N, T_0 = T, \Big(\tilde{\Gamma}_i\Big)_{i\ge 1}\right] = 0, \text{ if Assumption } \ref{ass:unc1a} \text{ holds}. \\
 (B) \quad  & \tau - \mathbb{E}\left[\hat{\tau}_t^v(w_v, v_v, \beta_v)\Big| N_0 = N, T_0 = T, \Big(\tilde{\Lambda}_t\Big)_{t\ge 1} \right]  = 0 ,  \text{ if Assumption } \ref{ass:unc2a} \text{ holds}.  \\
(C) \quad & \tau - \mathbb{E}\left[\hat{\tau}^{dr}(w_h, w_v, v_h, v_v)\Big| N_0 = N, T_0 = T \right]  = 0,  \text{ if  either Assumption } \ref{ass:unc1a} \text{ or } \ref{ass:unc2a} \text{ holds}. 
\end{aligned} 
$$ 
\end{prop} 

\begin{proof}[Proof of Proposition \ref{lem:expectations}]  The proof follows directly from Propositions \ref{lem:3a}, \ref{lem:verticala} and \ref{cor:2}. 
\end{proof}

Proposition \ref{lem:expectations} has important implications for inference. It shows that we can condition on the loadings for a horizontal regression and factors for a vertical regression, but \textit{not} vice-versa. Inference with the Synthetic DiD should take into account the randomness generated by \textit{both} the factors and loadings. We formalize these intuitions in the following section.

\begin{rem}[Heterogeneous treatment effects] 
\label{rem:heterogeneous}  Our analysis assumes that treatment effects are constant, i.e., $Y_{i,t}(1) = Y_{i,t}(0) + \tau$. Consider instead settings with heterogeneous (additive) treatment effects $\tau_{i,t} = Y_{i,t}(1) - Y_{i,t}(0)$. 
Effects heterogeneity may, for example, be relevant when treatment intensity varies by state due to different regulatory systems. Estimated treatment effects denote weighted averages, which depend on such heterogeneity.   In the context of Synthetic Controls (vertical regression), 
researchers may be able to identify a weighted combination of treatment effects $\sum_{t \ge T} q_{v, t} \sum_{n \ge N}  v_{v, n} \tau_{n, t}$ where the weights $v_{v, n}$ must satisfy the balancing restriction in Assumption \ref{ass:unc2a}. Intuitively, if we can guarantee that Assumption \ref{ass:unc2a} holds for $v_{v, n} = 1/N_1$, we can recover the average effect on the treated unit when choosing $q_{v,t} = 1/T$. In contrast, if Assumption \ref{ass:unc2a} holds for some other weighting mechanism, we can only recover a weighted combination of treatment effects. This intuition formalizes the trade-offs of matching restrictions on the weights $(v_h, v_v)$ and the identified class of estimands.
On the other hand, because $q_h, q_v$ can be chosen arbitrarily by the researcher, we recommend as choices for $q_h, q_v$, $q_h = 1/N_1 \mathbf{1}, q_v = 1/T_1 \mathbf{1}$, i.e., imposing equal weights on the estimators.\footnote{This choice is recommended under the failure of homogeneous treatment effects when the goal is to study the average effect on the treated unit. A second possible choice is to minimize the variance of the estimator, which, however, would require estimating the factors and loadings (and would affect the estimand of interest in the absence of homogeneous treatment effects).} \qed 
\end{rem}

\begin{table}[!htbp] \centering 
  \caption[Caption for LOF]{Summary of the identification strategies of different regressions. The ``Source of randomness" denotes the residual source of variation after conditioning on the endogenous variables (and that should be considered for confidence intervals presented in Section \ref{sec:inference_main}). ``Confounding effects" denote the components whose expectation is not zero, given the assignment mechanism. } 
  \label{tab:methods1} 
  \scalebox{0.65}{
\begin{tabular}{@{\extracolsep{5pt}} ccccc} 
\\[-1.8ex]\hline 
\hline \\[-1.8ex] 
Regression & Assumption & Source of randomness & Confounding Effects \\ 
\hline \\[-1.8ex] 
Horizontal 
 &  Restricted time confounding ($\tilde{\Lambda}_t \perp T_0$)   & $(\tilde{\Lambda}_t, \varepsilon_{i,t})$ & $(\lambda_t, \gamma_i + \tilde{\Gamma}_i)$  \\ 
 & + Exact matching of factors' conditional expectations & \\
& $(\sum_{s \le T} w_{h,s} \lambda_s = \sum_{t \ge T} v_{h,t} \lambda_t$ for some $w_h, v_h)$ &  
 \\
 & & \\ 
Vertical 
 &  Restricted unit confounding ($\tilde{\Gamma}_i \perp N_0$) & $(\tilde{\Gamma}_i,  \varepsilon_{i,t})$ & $(\gamma_i, \lambda_t + \tilde{\Lambda}_t)$ \\ 
 & + Exact matching of loadings' conditional expectations  & \\
 &$(\sum_{i \le N} w_{v,i} \gamma_i = \sum_{j \ge N} v_{v,j} \gamma_j$ for some $w_v, v_v$) & \\
& & \\ 
Synthetic DiD
 &  \textit{Either} no high dimensional unit & $(\tilde{\Gamma}_i, \tilde{\Lambda}_t, \varepsilon_{i,t})$ &  Either $(\lambda_t + \tilde{\Lambda}_t, \gamma_i)$  \\ 
 & or time confounders (and matching over time or units)  & &  or $(\lambda_t, \tilde{\Gamma}_i + \gamma_i)$ \\
 & & \\ 
PCA/Least Squares
 &  No high rank confounders & $\varepsilon_{i,t}$ & $(\tilde{\Gamma}_i + \gamma_i, \tilde{\Lambda}_t + \lambda_t)$ \\ 
 &  & \\

\hline \\[-1.8ex] 
\end{tabular} 
}
\end{table}

\section{Inference over the post-treatment period} \label{sec:inference_main}

This section studies inference for vertical, horizontal regression and Synthetic DiD. We consider the estimators in Equation \eqref{eqn:estimators2} for the horizontal and vertical regression and $\hat{\tau}^{dr}(\cdot)$ in Equation \eqref{eqn:estimators1} for the Synthetic DiD. 
We will condition on the event 
$
N_0 = N, T_0 = T$. We study asymptotics as  $N, T, N_1, T_1 \rightarrow \infty,   
$ 
with $N_1, T_1$ growing at an appropriate slower rate than $N, T$. We return to settings with a finite number of treated units in Appendix \ref{sec:placebo}. 

We study inference given some known (not data-dependent) weights
 $(w_h^\star, v_h^\star,\beta_h^\star), (w_v^\star, v_v^\star,\beta_v^\star)$. We will assume that these weights satisfy the restrictions in Assumptions \ref{ass:unc1a} and \ref{ass:unc2a}, respectively, and regularity conditions presented below.  
Corollary \ref{lem:aa} shows that estimation in the absence of estimation error is asymptotically equivalent to inference with estimated weights -- assuming that the post-treatment period is sufficiently shorter than the pre-treatment period. 
We study the estimation of the weights in Section \ref{sec:estimation}. 

As shown in Proposition \ref{lem:expectations}, inference must account for different sources of randomness depending on the nature of confounding: estimators are unbiased only unconditional (but not necessarily conditional) on either time or unit-level confounders, depending on the confounding assumption. We formalize this intuition below.

\subsection{Inference with horizontal and vertical regression} \label{sec:inference1}

We present assumptions and results for the horizontal regression first.

\begin{ass}[Horizontal regression conditions] \label{ass:horizontal_inference} Assume the following: 
\begin{itemize} 
\item[(A)] $T_1 N_1/T_0^{1/3}= o(1)$; 
\item[(B)] $||w_h^\star||_{\infty} = \mathcal{O}(T_0^{-2/3}), ||v_h^\star||_{\infty} = \mathcal{O}(T_1^{-2/3})$, and $w^\star, v^\star$ satisfy condition (B) in Assumption \ref{ass:unc1a}. 
\end{itemize} 
\end{ass}

Condition (A) states that the post-treatment period is much shorter than the pre-treatment period. Condition (B) assumes that the oracle weights guarantee pre and post treatment balance in the factors. Condition (B) assumes that no single (or few) units receive all the weights similar to \cite{arkhangelsky2018role}, ruling out sparse settings.



\begin{thm} \label{thm:horizontal_asym} Suppose that Assumptions \ref{ass:outcomea}, \ref{ass:inda}, \ref{ass:iid}, \ref{ass:unc1a},  \ref{ass:horizontal_inference} hold. Then as $T_1 \rightarrow \infty$, for any $q_h: \sum_{n \ge N} q_{h,n} = 1$, 
$$
\small 
\begin{aligned} 
 \frac{ \Big(\hat{\tau}^h(w_h^\star, v_h^\star, \beta_h^\star, q_h) - \tau\Big)}{\bar{\mathbb{V}}_h^{1/2}(N, T, q_h)}  \rightarrow_d \mathcal{N}(0,1)
\end{aligned} 
$$ 
where $\bar{\mathbb{V}}_h(N, T, q_h) = \mathbb{V}\left(  \sum_{n \ge N} q_{h, n} \sum_{t \ge T} v_{h,t}^\star Y_{n,t}  \Big| N_0 = N, T_0 = T, \Big(\tilde{\Gamma}_i\Big)_{i\ge 1}\right)$.
\end{thm} 

\begin{proof} See Appendix \ref{proof:thm_hor1}.  
\end{proof}

Theorem \ref{thm:horizontal_asym} shows that the horizontal estimator, \textit{for any} weighted combination of treated units converges in distribution to a Gaussian. The convergence rates also depend on unit-level confounders. The variance can be estimated directly using the empirical moments of the post-treatment outcomes.

\begin{cor}[Rate of convergence] \label{cor:loadings} Let the conditions in Theorem \ref{thm:horizontal_asym} hold. Then 
$$
\small 
\begin{aligned} 
\frac{\bar{\mathbb{V}}_h(N, T, q_h)}{||v_h^\star||_2^2} =   \Big(\sum_{n \ge N} (\tilde{\Gamma}_n + \gamma_n) q_{h, n}\Big)^\top \Sigma_{\tilde{\Lambda}}  \Big(\sum_{n \ge N} (\tilde{\Gamma}_n + \gamma_n) q_{h, n}\Big) + \sigma_{\varepsilon}^2 ||q_h||_2^2, 
\end{aligned} 
$$ 
where $\Sigma_{\tilde{\Lambda}} = \mathbb{V}(\tilde{\Lambda}_t), \sigma_\varepsilon^2 = \mathbb{V}(\varepsilon_{i,t})$. 
\end{cor} 
Corollary \ref{cor:loadings} shows that the variance (and rate of convergence) depends on the norm of the weights and on the loadings. The rate of convergence is weakly slower than $\frac{1}{\sqrt{T_1 N_1}}$, and equal to $\frac{1}{\sqrt{T_1 N_1}}$ if the \textit{loadings} are exogenous. Specifically, we can write 
$$
\bar{\mathbb{V}}_h(N,T, q_H) = \mathcal{O}\Big(||q_h||^2 || v_h^\star||_2^2 \sigma_\varepsilon^2 + \rho_{N_1}(q_h) ||v_h^\star||_2^2 \Big), \quad \rho_{N_1}(q_h) = \Big| \Big|\sum_{n \ge N} q_{h,n} (\tilde{\Gamma}_n + \gamma_n) \Big| \Big|_2^2. 
$$  
Here, $||v_h^\star||_2^1 = 1/T$, and $\rho_{N_1}$ characterizes the strength of the unit-level confounding, with $\rho_{N_1}$ either bounded away from zero, or $\rho_{N_1} \rightarrow 0$.  For $\rho_{N_1} \rightarrow 0$ to converge to zero, we would need no unit-level confounding, that is $\frac{1}{N} \sum_{n \ge N} \tilde{\Gamma}_n$ concentrates around its unconditional expectation (unconditional on $N_0$) and $\gamma_n$ is local to zero.  This result illustrates properties of convergence rates that depend on the strength of confounding.



The Synthetic Control case follows similarly to the horizontal regression case, where we invert the role of the loadings and factors, with 
$$
\small 
\begin{aligned} 
 \frac{\Big(\hat{\tau}^v(w_v^\star, v_v^\star, \beta_v^\star, q_v) - \tau\Big)}{\bar{\mathbb{V}}_v^{1/2}(N, T, q_v)}  \rightarrow_d \mathcal{N}(0,1)
\end{aligned} 
$$ 
where $\bar{\mathbb{V}}_v(N, T, q_v) = \mathbb{V}\left( \sum_{t \ge T}  q_{v, t}  \sum_{n \ge N} v_{v, n}^\star Y_{n,t} \Big| N_0 = N, T_0 = T, \Big(\tilde{\Lambda}_t\Big)_{t\ge 1}\right)$. Here, 
$$
\begin{aligned} 
\frac{ \bar{\mathbb{V}}_v(N, T, q_v)}{||v_v^\star||_2^2} =  \Big(\sum_{t \ge T} (\tilde{\Lambda}_t + \lambda_t) q_{v, t}\Big)^\top \Sigma_\Gamma  \Big(\sum_{t \ge T} (\tilde{\Lambda}_t + \lambda_t) q_{v, t}\Big) + \sigma_\varepsilon^2 ||q_v||_2^2,  
\end{aligned} 
$$ 
where $\Sigma_{\tilde{\Gamma}} = \mathbb{V}(\tilde{\Gamma}_i)$. With diagonal $\Sigma_{\tilde{\Gamma}}$, the rate of convergence is of order $||v_v^\star||_2^2 ||q_v||_2^2 \sigma_\varepsilon^2 + ||v_v^\star||_2^2 \Big| \Big| \sum_{t \ge T} q_{v,t} (\tilde{\Lambda}_t + \lambda_t)\Big| \Big|_2^2$, which depends on time-level confounders.  
Therefore, whereas each regression provides consistent estimates as $N, T \rightarrow \infty$ under their corresponding unconfoundedness restriction (\textit{either} no high dimensional time or unit level confounders), a \textit{faster} rate of convergence can be obtained if both no high dimensional time \textit{and} unit level confounders hold. 

We conclude this discussion by showing how our results in Theorem \ref{thm:horizontal_asym} apply in the presence of estimated weights.

\begin{cor} \label{lem:aa} Suppose that the conditions in Theorem \ref{thm:horizontal_asym} hold, and consider any data dependent weights $(\hat{w}_h, \hat{v}_h, \hat{\beta}_h)$ such that 
\begin{equation} \label{eqn:error} 
\hat{\tau}_{n,t}^h(\hat{w}_h, \hat{v}_h, \hat{\beta}_h) - \hat{\tau}_{n,t}^h(w_h^\star, v_h^\star, \beta_h^\star) = o_p\left(N_1^{-1/2} T_1^{-1/2}\right).
\end{equation}  
Then Theorem \ref{thm:horizontal_asym} holds with  $\hat{\tau}^h(\hat{w}_h, \hat{v}_h, \hat{\beta}_h)$ in lieu of  $\hat{\tau}^h(w_h^\star, v_h^\star, \beta_h^\star)$. 
\end{cor} 

\begin{proof}[Proof of Corollary \ref{lem:aa}] See Appendix \ref{proof:lemma_convergence}.  
\end{proof}


An important insight of Corollary \ref{lem:aa} is that even if the variance is conditional on the loadings for the horizontal regression, the estimated weights only need to converge to their population counterpart unconditionally on the factors and loadings. The reason is that the variance $\bar{\mathbb{V}}_h$ converges to zero at a rate at most $1/\sqrt{N_1 T_1}$, while the estimated weights can converge at a faster rate when the pre-treatment period and number of control units is larger than the post-treatment period and number of treated units. Appendix \ref{proof:lemma_convergence} presents the details. Conditions on the convergence rates of the weights are discussed in Section \ref{sec:estimation}, where we review different weights estimators.

 \begin{rem}[Inference with a single treated unit] In Appendix \ref{sec:placebo}, we study inference in the presence of a single treated unit. We show how placebo tests in \cite{abadie2010synthetic} can be used in our setting under additional unconfoundedness restrictions on the endogenous loadings, assuming that we can find a donor pool of ``placebo treated units" such that we can exactly match their loadings with the loadings of the remaining control units.  \qed 
 \end{rem}

\subsection{Robust inference}

In the following lines, we study inference that is robust to confoundedness of either the loadings or factors. We impose the following conditions. 

\begin{ass}[Robust regression conditions] \label{ass:robust_inference} Assume the following: 
\begin{itemize} 
\item[(A)] $T_1 N_1/T_0^{1/3} = o(1), N_1 T_1/N_0^{1/3} = o(1)$; 
\item[(B)] $||w_h^\star||_{\infty} = \mathcal{O}(T_0^{-2/3}), ||w_v^\star||_{\infty} = \mathcal{O}(N_0^{-2/3})$, $||v_h^\star||_{\infty} = \mathcal{O}(T_1^{-2/3}), ||v_v^\star||_{\infty} = \mathcal{O}(N_1^{-2/3})$, and either $(w_h^\star, v_h^\star)$ satisfy (B) in Assumption \ref{ass:unc1a} or $(w_v^\star, v_v^\star)$ satisfy (B) in Assumption \ref{ass:unc2a}. 
\end{itemize} 
\end{ass} 

Assumption \ref{ass:robust_inference} formalizes double robustness properties, for which either no high dimensional unit confounders or no high dimensional time confounders exist.

\begin{thm} \label{thm:dr} Let Assumptions \ref{ass:outcomea}, \ref{ass:inda}, \ref{ass:iid}, \ref{ass:robust_inference} hold, and either Assumption \ref{ass:unc1a} or Assumption \ref{ass:unc2a} (or both) hold. Then as $N_1, T_1 \rightarrow \infty$, 
$$
\small 
\begin{aligned} 
P\left(\frac{\Big(\hat{\tau}^{dr}(w_h^\star, w_v^\star, v_h^\star, v_v^\star) - \tau\Big)}{\max\{\mathbb{V}_h(v_h^\star, v_v^\star), \mathbb{V}_v(v_v^\star, v_h^\star)\}^{1/2}} \le z_{1 -\alpha} \Big| N_0 = N, T_0 = T\right) \ge 1 - \alpha,
\end{aligned} 
$$ 
where $\Phi(z_{1-\alpha}) = 1- \alpha$, $\Phi$ is the Gaussian CDF, and 
$$
\small 
\begin{aligned} 
\mathbb{V}_h(v_h^\star, v_v^\star) & = \mathbb{V}\left(\sum_{t \ge T} \sum_{j \ge N} v_{h,j}^\star v_{v, t}^\star \Big\{ Y_{j,t} -  \sum_{i < N} w_{v, i}^\star Y_{i, t}\Big\}\Big| N_0 = N, T_0 = T, (\tilde{\Gamma}_i)_{i \ge 1}\right), \\  \mathbb{V}_v(v_v^\star, v_h^\star) & = \mathbb{V}\left( \sum_{j \ge N} \sum_{t \ge T} v_{v, t}^\star v_{h,t}^\star \Big\{ Y_{j,t} - \sum_{s < T} w_{v,s}^\star Y_{j, s}\Big\} \Big| N_0 = N, T_0 = T, (\tilde{\Lambda}_t)_{t \ge 1}\right). 
\end{aligned} 
$$ 
\end{thm} 

\begin{proof}[Proof of Theorem \ref{thm:dr}] 
 See Appendix \ref{proof:long}. 
 \end{proof}
Theorem \ref{thm:dr} shows that confidence intervals for Synthetic DiD depend on the \textit{worst}-\textit{case} variance conditional on either the factors or loadings. The variance calculation differs from settings with non-random factors and loadings  \citep[e.g.][]{arkhangelsky2021synthetic}, and can be computed using the empirical moments of weighted combinations of the outcomes.\footnote{Because the $\max\{\}$ operator is a continuous function, we can use the estimated variances $\mathbb{V}_h, \mathbb{V}_u$ and invoke Slutsky theorem to provide asymptotically valid confidence intervals. }  This difference is because of the population double-robustness property we derived.  

 A direct corollary of Theorem \ref{thm:dr} is that the estimator $\hat{\tau}^{dr}(\cdot)$'s convergence rate depends on both $N_1, T_1$, \textit{and} on the unconfoundedness restriction. The rate of converge is $\max\{\mathbb{V}_h(v_h^\star, v_v^\star), \mathbb{V}_v(v_v^\star, v_h^\star)\}$, namely (let $\wedge$ denote the maximum operator, $\Gamma_i = \gamma_i + \tilde{\Gamma}_i, \Lambda_t = \lambda_t + \tilde{\Lambda}_t$)
\begin{equation} \label{eqn:rhs}
\small 
\begin{aligned}
   ||v_v^\star||_2^2 ||v_h^\star||_2^2 \sigma_\varepsilon^2 + \Big\{&||v_h^\star||_2^2 \Big(\sum_{n \ge N} \Gamma_n  v_{v,n}^\star - \sum_{j < N} \Gamma_j w_{v,j}^\star\Big)^\top \Sigma_{\tilde{\Lambda}} \Big(\sum_{n \ge N} \Gamma_n v_{v,n}^\star - \sum_{j < N} \Gamma_j w_{v,j}^\star\Big), \\ & \quad  \wedge ||v_v^\star||_2^2  \Big(\sum_{t \ge T} v_{h, t}^\star \Lambda_t - \sum_{s <T} \Lambda_s w_{h, s}^\star\Big)^\top \Sigma_{\tilde{\Gamma}}  \Big(\sum_{t \ge T} v_{h, t}^\star \Lambda_t - \sum_{s <T} \Lambda_s w_{h, s}^\star\Big)\Big\}. 
\end{aligned} 
\end{equation}
If \textit{both} $\Lambda_t \perp T_0, \Gamma_i \perp T_0$, the right-hand side of Equation \eqref{eqn:rhs} converges almost surely to zero, with rate $1/\sqrt{N_1 T_1}$. On the other hand, under lack of either unconfoundedness restriction (over time or units), the right-hand-side expression in Equation \eqref{eqn:rhs} does not converge to zero and the rate of convergence is $\max\{T_1^{-1/2}, N_1^{-1/2}\}$. Convergence rates faster than $1/\sqrt{T_1}$ or $1/\sqrt{N_1}$ do not require that factors or loadings are exogenous. Convergence rates of order $1/\sqrt{N_1 T_1}$ allow the conditional expectations of the loadings and factors to be different from zero but require that the expectations match before and after the treatment \textit{after reweighting}.

\begin{table}[!h] \centering 
  \caption[Caption for LOF]{Convergence rates under the model in Assumption \ref{ass:outcomea} with independent factors and loadings. The convergence rate for Horizontal and Vertical regression depends on the post-treatment weights $v^\star$ and on the concentration of the loadings or factors. The convergence rate of Synthetic DiD depends on the post-treatment weights and the \textit{mismatch} between the post-treatment and pre-treatment factors and loadings.} 
  \label{tab:methods2} 
  \scalebox{0.66}{
\begin{tabular}{@{\extracolsep{5pt}} cccc} 
\\[-1.8ex]\hline 
\hline \\[-1.8ex] 
Regression &  Bound on Standard Error  \\ 
\hline \\[-1.8ex] 
Horizontal
 & $ ||v_h^\star||_2 \Big| \Big| \sum_{n \ge N} q_{h,n} (\tilde{\Gamma}_n + \gamma_n) \Big| \Big|_2$ for arbitrary $q_h: 1^\top q_h = 1$   \\ 
 &  & \\
& &  
 \\
Vertical
 &  $ ||v_v^\star||_2 \Big| \Big| \sum_{t \ge T} q_{v,n} (\tilde{\Lambda}_t + \lambda_t) \Big| \Big|_2$ for arbitrary $q_v: 1^\top q_v = 1$   \\ 
 &  & \\
 & & \\

Synthetic DiD
 & $  ||v_h^\star||_2 \Big|\Big|\sum_{n \ge N} (\tilde{\Gamma}_n + \gamma_n) v_{v,n}^\star - \sum_{j < N} (\tilde{\Gamma}_j + \gamma_j) w_{v,j}^\star\Big|\Big|_2 \wedge  ||v_v^\star||_2 \Big|\Big| \sum_{t \ge T} v_{h, t}^\star (\tilde{\Lambda}_t + \lambda_t) - \sum_{s <T} (\tilde{\Lambda}_s + \lambda_s) w_{h, s}^\star\Big|\Big|_2$ \\ 
 &  & \\
 & & \\

\hline \\[-1.8ex] 
\end{tabular} 
}
\end{table}

\section{Estimation of the weights and factors: a discussion} \label{sec:estimation}

In this section, we review existing methods for estimating the weights and discuss properties and assumptions that these methods require in the context of our model of confounding.   




\subsection{Weights estimation with penalized $l_2$-norm minimization} \label{sec:squared_loss}

First, we review estimation of the weights as in \cite{arkhangelsky2021synthetic}, which are similar to those in \cite{abadie2010synthetic} with the additional norm constraint. For the horizontal regression, we can construct weights' estimators as follows: 
\begin{equation} \label{eqn:est1}
\small 
\begin{aligned} 
(\hat{\beta}_h, \hat{w}_h, \hat{v}_h) & = \mathrm{arg}\min_{w,  v, \beta, 1^\top w  = 1, w \ge 0, ||w||_{\infty} \le T_0^{-2/3}, 1^\top v = 1, v \ge 0, ||v||_{\infty} \le T_1^{-2/3} } \mathcal{L}_{h, N, T}(w, v, \beta)  , \\
\mathcal{L}_{h, N, T}(w, v, \beta) & =  \sum_{i < N} \Big( \sum_{t \ge T} v_{h, t} Y_{i,t} - \sum_{s < T} w_s Y_{i,s} - \beta\Big)^2 + p_h N  (||w||^2 + ||v||^2) 
\end{aligned} 
\end{equation} 
where we minimize the $l_2$-distance between the post-treatment average control outcomes and a weighted combination of the pre-treatment outcomes. 
Similarly, for the vertical regression  
\begin{equation} \label{eqn:est2} 
\small 
\begin{aligned} 
(\hat{\beta}_v, \hat{w}_v, \hat{v}_v) & = \mathrm{arg}\min_{w,  v, \beta, 1^\top w  = 1, w \ge 0, ||w||_{\infty} \le N_0^{-2/3}, 1^\top v = 1, v \ge 0, ||v||_{\infty} \le N_1^{-2/3} }\mathcal{L}_{v, N,T}(w', v', \beta') \\ 
\mathcal{L}_{v, N,T}(w', v', \beta') & = \sum_{s < T} \Big( \sum_{i \ge N} v_{v, i}' Y_{i,s} - \sum_{i < N} w_i' Y_{i,s} - \beta'\Big)^2 +  p_v T (||w'||^2 + ||v'||^2) .
\end{aligned}
\end{equation}   Here, $p_v, p_h$ denote choosen penalty parameters discussed in \cite{arkhangelsky2021synthetic}. 

The constraints are as discussed in Sections \ref{sec:identification}, \ref{sec:inference_main}.  Algorithm \ref{alg:2} presents a summary. 

  \begin{algorithm} [!h]   \caption{Vertical and Horizontal Regressions}\label{alg:2}
    \begin{algorithmic}[1]
    \Require $Y_{i,t}, i < N_0, t \in \{1, \cdots, T\}$
    \State Minimize $\mathcal{L}_{h, N_0,T_t}(w, v, \beta), \mathcal{L}_{v, N_0,T_0}(w', v', \beta')$ under the constraints  
    $$
    \small 
    \begin{aligned}
     \text{ subject to }  & 1^\top w = 1,  ||w||_{\infty} \le K N_0^{-2/3}, w \ge 0, 1^\top v = 1, ||v||_{\infty} \le K N_1^{-2/3}, v \ge 0
 \\  
& 1^\top w'  = 1, ||w'||_{\infty} \le K T_0^{-2/3}, w' \ge 0, 1^\top v = 1, ||v'||_{\infty} \le K T_1^{-2/3}, v' \ge 0
    \end{aligned}
    $$ 
    \Return $(\hat{w}_h, \hat{w}_v, \hat{v}_h, \hat{v}_v)$. 
         \end{algorithmic}
\end{algorithm}

The minimization problem in Section \ref{sec:squared_loss} corresponds to a error-in-variables optimizaton problem. By letting $\Gamma_i = \gamma_i + \tilde{\Gamma}_i$ 
\begin{equation} \label{eqn:error_in_model} 
\small 
\begin{aligned} 
& \sum_{i < N_0}   \left(\sum_{t \ge T_0} v_{h,t} Y_{i,t} - \sum_{s < T} w_{h,t} Y_{i,s}\right)^2 + (\sigma^2_h - 1) \zeta_h N_0  (||w||^2 + ||v||^2)  = \\ & 
\sum_{i < N_0}  \left(\sum_{t \ge T_0} v_{h,t} \Big( \lambda_t^\top \Gamma_i + \eta_{i,t} + \iota_{1, t}\Big) - \sum_{s < T_0} w_{h,s} \Big(\lambda_s^\top \Gamma_i +  \eta_{i,s} + \iota_{1,s} + \beta\Big)\right)^2 \\ & \quad \quad \quad \quad +   (\sigma^2_h - 1) \zeta_h N_0 (||w||^2 + ||v||^2) , 
\end{aligned} 
\end{equation} 
where 
\begin{equation}
\eta_{i,t}  = \tilde{\Lambda}_t^\top \Gamma_i + \varepsilon_{i,t}, \quad \Gamma_i = \gamma_i + \tilde{\Gamma}_i
\end{equation}

Following \cite{arkhangelsky2021synthetic}, we define ``oracle" counterpart of such a problem for the horizontal regression solves the following optimization problem.  
\begin{equation} \label{eqn:obj_oracles} 
\small 
\begin{aligned} 
& (w_h^\star, v_h^\star,  \beta_h^\star) = \mathrm{arg} \min_{w,  v, \beta, 1^\top w  = 1, w \ge 0, ||w||_{\infty} \le T_0^{-2/3}, 1^\top v = 1, v \ge 0, ||v||_{\infty} \le T_1^{-2/3} } \mathcal{L}_h^\star(w,  v, \beta, \delta_h),  \\
& \mathcal{L}_h^\star(\cdot) = \sum_{i < N_0} \delta_{h, i} \left(\Gamma_i^\top\Big(\sum_{t \ge T_0} v_t \lambda_t - \sum_{s < T_0} w_s \lambda_s \Big) + \sum_{t \ge T_0} v_{t} \iota_{1, t} - \sum_{s < T_0} w_{s} \iota_{1, s} - \beta\right)^2 +  N_0 p_h^\star  (||w||^2 + ||v||^2).
\end{aligned} 
\end{equation} 
where $p_h^\star$ is a penalty that depends on the residuals' variance as in \cite{arkhangelsky2021synthetic}. The penalization of the oracle solution is different from the penalization used for estimation as motivated in \cite{hirshberg2021least}.  
We can define oracle weights for Synthetic Control $(w_v^\star, v_v^\star, \beta_v^\star)$ in a similar manner.

Consider estimating the horizontal regression weights, and 
let $||\cdot||_{op}$ be the operator norm for a matrix,\footnote{For a matrix $A$ the operator norm is defined as $\sup_{u, ||u|| = 1} ||A u||$.} and 
$$
\begin{aligned} 
\Sigma_\eta & = \frac{1}{T_0 + T_1} \mathbb{E}\Big[\tilde{\eta} \tilde{\eta}^\top \Big| (\tilde{\Gamma}_i)_{i \ge 1}, N_0 = N, T_0  = T\Big], \quad 
\mu^2 = ||\Sigma_\eta||_{op}. 
\end{aligned} 
$$ 
Under Assumptions  \ref{ass:outcomea}, \ref{ass:inda}, \ref{ass:iid}, \ref{ass:unc1a}
\begin{equation} \label{eqn:matrix1} 
\begin{aligned} 
\Sigma_\eta^{(i,j)} 
= \begin{cases} 
 \Big(\tilde{\Gamma}_i + \gamma_i\Big)^\top \mathrm{Var}(\tilde{\Lambda}_t) \Big(\tilde{\Gamma}_i + \gamma_i\Big) + \mathrm{Var}(\varepsilon_{i,t}), \quad & i = j \\ 
 \Big(\tilde{\Gamma}_i +\gamma_i\Big)^\top \mathrm{Var}(\tilde{\Lambda}_t) \Big(\tilde{\Gamma}_j + \gamma_j\Big) \quad & \text{otherwise}. 
\end{cases} 
\end{aligned} 
\end{equation} 

It follows that conditional the treated units $N_0 = N$, $\mu^2$ measures the amount of endogeneity of $\tilde{\Gamma}_i$. Specifically, let  
$$
\small 
\begin{aligned} 
  A & \in \mathbb{R}^{ (N_0 - 1) \times (T_0 + T_1) }, \quad A_{i,t} =  
  \begin{cases}
 \lambda_t^\top \Gamma_i + \iota_{0, t} & \text{ if } t  < T_0 \\ 
  - \lambda_t^\top \Gamma_i - \iota_{0,t} & \text{ otherwise}. 
  \end{cases}, \\ \tilde{\eta} & \in \mathbb{R}^{(N_0  - 1) \times (T_0 + T_1)}, \quad  \tilde{\eta}_{i,t} = \begin{cases} 
   \eta_{i,t}, & \text{ if } t < T_0 \\ 
   - \eta_{i,t} , & \text{ otherwise}.   
\end{cases}  
\end{aligned} 
$$ 
Here the matrix $A$ depends on the endogenous factors; $\tilde{\eta}$ is the noise matrix for the \textit{control} units. Define 
$$
\theta_h = (w_h, v_h)  \text{ and } \theta_h^\star = (w_h^\star, v_h^\star).
$$


In Appendix \ref{app:convergence}, using rate-properties of error in variable models in \cite{hirshberg2021least} applied to our model of confounding, we show that the convergence rate depends on three main components: 
$$
  \frac{\mathrm{rank}(A) \mu^2 ||\theta^\star||_2^2}{N_0},  \quad \mu^2 ||\theta^\star||_2 \log(T_0 + T_1) N_0^{-1/2}, \quad \frac{\mu^2 \log(T_0 + T_1)}{N_0}. 
$$ 
Importantly, here, the rank of the matrix $A$ only depends on the non-zero number of endogenous factors $\lambda_t$, i.e., $||\lambda_t||_0$. In the presence of few endogenous factors, the rank of $A$ is small, even in settings where $\tilde{\Lambda}_t$ is high-dimensional. The component $\mu^2$ depends on the degree of endogeneity of the \textit{loadings} (unit-level confounders). The rate of convergence of $\mu^2$ is of order $\sqrt{N_0}$ from standard properties of matrix concentration inequalities \citep{van2017spectral} for exogenous $\tilde{\Gamma}_i$ and slower otherwise. The component $||\theta_h^\star||$ instead converges to zero as $T_1 \rightarrow \infty$. The error does not necessarily converge to zero for \textit{arbitrarly} endogenous loadings $\tilde{\Gamma}_i$. It converges to zero under restrictions on $\mu^2$ (e.g., only few loadings $\tilde{\Gamma}_i$ are endogenous and the remaining ones are exogenous). The same results holds for vertical weights once we exchange the role of the factors and loadings.

This result illustrates the benefits of the Synthetic control methods in the presence of high-rank factor models.  

\begin{rem}[Balancing]
We can gain further intuition if we interpret the weights' estimator in Section \ref{sec:squared_loss} as the dual of a balancing problem. For given constraint $\nu = o(1)$ (as a function of $N, T$) we can formulate the dual of Equation \eqref{eqn:est1} as 
$$
\small 
\begin{aligned} 
\min_{w_h, v_h, \beta} ||w_h||^2 + ||v_h||_2^2, \quad \text{ s.t. } & \underbrace{\frac{1}{N} \sum_{i < N} \Big(\sum_{t \ge T_0} v_{h,t} Y_{i,t} - \sum_{s > T} w_{h,s} Y_{i,s} + \beta\Big)^2}_{:=\frac{1}{N} ||(A + \tilde{\eta}) \theta_h + \beta||_2^2}  \le \nu, \quad (w_h, v_h) \ge 0, \\ & ||w_h||_{\infty} \le T_0^{-2/3}, ||v_h||_{\infty} \le T_1^{-2/3}, 1^\top w_h = 1^\top v_h = 1. 
\end{aligned} 
$$ 
The dual formulation clarifies the role of the constraint on the weights $w_h$: The constraint on the weights norm guarantees the variance of the error component $\tilde{\eta} \theta_h$ converges to zero (its variance is of order $||\theta_h||^2$). Intuitively, the Synthetic Control method averages over the component $\tilde{\eta}$ to obtain consistent weight's estimators without estimating the factors directly.\footnote{It is possible to consider alternative balancing estimators. For example, we might balance the means of treated and control units by imposing 
$$
\Big|\frac{1}{N} \sum_{i < N} \sum_{t \ge T} v_{h,t} Y_{i,t} - \frac{1}{N} \sum_{i < N} \sum_{s > T} w_{h,s} Y_{i,s} \Big | \le \nu', 
$$
for some constraint $\nu' = o(1)$. Imposing such balancing restriction assumes that we can find a set of weights $v_h^\star, v_h^\star$ such that $\sum_{t < T} v_{h,t}^\star \lambda_t = \sum_{t \ge T} w_{h,t}^\star \lambda_t$ \textit{and} $\sum_{t < T} v_{h,t}^\star \iota_{1,t} = \sum_{t \ge T} w_{h,t}^\star \iota_{1, t}$, i.e., we can also match fixed effects before and after the treatment.} \qed 
\end{rem}

\subsection{Proximal methods for weights estimation}

It is possible to use proximal methods to estimate the weights in the spirit of \cite{shi2021theory}, \cite{imbens2021controlling}.
Consider estimating synthetic control weights first. Suppose we can find variables $Z_t$ independent of $\varepsilon_{i,t}$ such that for all $w_{v}, v_v$
\begin{equation} \label{eqn:proximal}
\mathbb{E}\Big[Z_t \Big(\sum_{j < N} w_{v,j} \tilde{\Gamma}_j^\top - \sum_{i \ge N} v_{v, i} \tilde{\Gamma}_i\Big)^\top \tilde{\Lambda}_t \Big| N_0 = N, T_0 = T, \tilde{\Lambda}_t\Big] = 0, \quad \forall t \le T_0. 
\end{equation} 
Assuming that the number of such variables $Z_t$ is sufficiently larger than the number of  parameters (weights) to be estimated, we can use such moment restrictions to estimate the weights. The main advantage of the proximal method is that it does not impose restrictions on the distribution of the endogenous factors $\tilde{\Lambda}_t$. However, it requires finding a set of proximal variables that satisfy Equation \eqref{eqn:proximal}. A similar approach follows for horizontal weights, where we should find a (recenter) instrument $X_i$ such that 
$$
\mathbb{E}\Big[X_i \Big(\sum_{t < T} w_{h,t} \tilde{\Lambda}_t^\top - \sum_{s \ge T} v_{h, s} \tilde{\Lambda}_s\Big)^\top \tilde{\Gamma}_i \Big| N_0 = N, T_0 = T, \tilde{\Gamma}_i\Big] = 0, \quad \forall i \le N_0. 
$$

\subsection{Factor models estimators} \label{sec:factors}

We conclude this section with a discussion on factor models. We write our model as (for $\Gamma_i = \tilde{\Gamma}_i + \gamma_i$)
 \begin{equation} \label{eqn:model_working}
 Y_{i,t}(0) = \lambda_t^\top \Gamma_i + \underbrace{\tilde{\Lambda}_t^\top \Gamma_i  + \varepsilon_{i,t}}_{= \eta_{i,t}}, 
\end{equation}  
where $\lambda_t^\top \Gamma_i$ is low rank when $||\lambda_t||_0$ is finite, and $\tilde{\Lambda}_t^\top \Gamma_i$ can be an high rank component.

Our discussion above suggests that horizontal regression (and similarly vertical regression after exchanging the factors with the loadings)   leverage the assumption that endogenous factors $\lambda_t$ are low dimensional (e.g., $||\lambda_t||_0$ is uniformly bounded).

It is interesting to contrast the convergence rate of the weights of Synthetic controls (described in Section \ref{sec:squared_loss} and formally presented in Proposition \ref{proof:estimation_error}) with those that we would obtain using standard methods for estimating factor models. Two main differences from properties of least squares estimators as in
\cite{bai2009panel} is that Synthetic Control does not require a (i) low-rank assumption on the factors, but only a low-rank assumption on the \textit{endogenous} factors; (ii) it does not require to specify or estimate the number of (endogenous) factors. 
For example, if we consider a factor model where the factors are aggregate sectoral shocks and the loadings are sectoral weights, a sparsity assumption on loadings implies that the policy implementation may only depend on a subset of sectoral weights.

More closely related to the results in Proposition \ref{proof:estimation_error},  \cite{moon2017dynamic} show that conditions for least squares factor model can be expressed as a function of the operator norm ($||\eta||_{op}$). The main distinction from Proposition \ref{proof:estimation_error}, however, is that the number of factors (the degree of sparsity of $\lambda_t$ in our framework) must be known for \cite{moon2017dynamic}'s results to hold. With an unknown number of factors (as in the case of Proposition \ref{proof:estimation_error}), stronger conditions on the error term $\eta$, such as rank restrictions, are imposed \citep[][]{moon2015linear}. 

The model in Equation \eqref{eqn:model_working} also connects to the approximate factor model discussed in \cite{chamberlain1982arbitrage}. The main distinction, however, is that low dimensional factor structure $\lambda_t^\top \gamma_i$ cannot be necessarily separated by the (exogenous and high dimensional) structure 
$\tilde{\Gamma}_i^\top \tilde{\Lambda}_t$, different from \cite{chamberlain1982arbitrage}.  
This is an important distinction also from work in the literature of causal inference with panel data as in \cite{athey2021matrix} (see, e.g., Section 8.2). 

Interestingly, Equation \eqref{eqn:model_working} may justify alternative estimators \textit{if} additional conditions are imposed on the loadings and factors. For example, under low rank $\lambda_t^\top \Gamma_i$ and sparse  $\tilde{\Lambda}_t^\top \Gamma_i$ we could use estimators in \cite{candes2011robust}, and for bounded eigenvalues of $\tilde{\Lambda}_t^\top \Gamma_i$ as in \cite{bai2019rank}. 

Finally, 
Table \ref{tab:1} presents numerical studies that support the benefit of Synthetic control methods over least square regression.

\begin{table}[!h] \centering 
  \caption{Illustrative example where we report the mean-squared error using data from CPS on state earnings \citep{arkhangelsky2021synthetic}, after removing the treated units. Here, $T_0 = 30, T_1 = 10, N= 42$. We simulate an environment where either two or three units (originally under control in the original study) are exposed to treatment, with treatment effects equal to zero. We compute the root-mean-squared error of the estimated average treatment effect for each estimator, after averaging the mean-squared error averaged over all pairs or triads of placebo treated units in the sample. Here, vertical denotes a vertical regression and similarly horizontal for the horizontal regression. PCA denotes a regression where we first estimate the principal components via PCA using all control units. We then run a regression for each placebo-treated unit on the principal components and compute the counterfactuals. The weights for Synthetic DiD, vertical and horizontal regression are all computed as in Section \ref{sec:estimation}, with a small penalty $\lambda = 0.01$ to guarantee stability, a constraint that sum to one, and constraints on the $l_{\infty}$-norm as in Section \ref{sec:estimation}. For PCA, the number of factors is computed using the BIC.  }
  \label{tab:1} 
\begin{tabular}{@{\extracolsep{5pt}} ccc} 
\\[-1.8ex]\hline 
Root-mean-squared-error & $N_1 = 3$ & $N_1 = 2$ \\ 
\hline \\[-1.8ex] 
Synthetic DiD & $0.020$ & $0.025$ \\ 
Vertical & $0.020$ & $0.025$ \\ 
Horizontal & $0.026$ & $0.031$ \\ 
PCA & $0.674$ & $0.674$ \\ 
\hline \\[-1.8ex] 
\end{tabular} 
\end{table}

\section{Conclusion}

This paper studies inference on treatment effects in panel data settings in the presence of confounding. We model confounding through unobserved factors -- which might affect when the treatment occurs -- and unobserved loadings -- which might affect which units receive the treatment. 
We illustrate the existence of a trade-off between assuming no (high dimensional) confounding across units or time and introducing notions of double robustness in this setting. We relate notions of confounding to the source of randomness for confidence intervals.

This paper opens new questions on trade-offs between the choice of the estimator and robustness to confounding. Different sources of confounding justify different estimators, including Synthetic DiD, Synthetic Control, factor models, or proxy variable methods. A comprehensive comparison of such estimators remains an open question.  
Future research should also study trade-offs between weak factors and the choice of the estimator. Synthetic control methods implicitly leverage the low-rank representation of the confounders, whereas the least-squares method estimates all such confounders. Finally, a further avenue for future avenue of research is to study augmented inverse probability weight estimators in the spirit of  \cite{robins1994estimation} in contexts with synthetic controls and high dimensional factor models.

 \bibliography{my_bib2}
\bibliographystyle{chicago}

\appendix

\section{Randomization inference with a single treated unit} \label{sec:placebo}


In this section, we revisit placebo tests in \cite{abadie2010synthetic} and illustrate how such tests can be used for inference. Consider conducting inference under the null hypothesis 
$$
H_0: \tau = 0, 
$$ 
which corresponds to a sharp null hypothesis by the assumption of additive and homogeneous treatment effects. We construct a set of ``placebo" treated units, where we pool with the treated unit few control units ``as if" they were treated.\footnote{The larger the set of placebo-treated units, the better the asymptotic approximation through randomization inference at the expense of reducing the size of the control pool. We study inference with few treated units for Synthetic Controls.} Algorithm \ref{alg:placebo} describes the procedure.

  \begin{algorithm} [!h]   \caption{Randomization inference for Synthetic DiD with a single treated unit}\label{alg:placebo}
    \begin{algorithmic}[1]
    \Require $N_1$ (number of placebo treated units), $N$ (number of control units after excluding the placebo treated units)
    \State Estimate a placebo treatment effect estimator for each unit $j \ge N$ as in Equation \eqref{eqn:placebo}, after excluding the other units in the same set $j \ge N$ when estimating the placebo treatment effect. 
    \State Construct confidence intervals by permuting which unit receives the treatment in the set $j \ge N$ and using the empirical distribution obtained from such permutations (after excluding the treated unit). \\
    \Return P-value for testing the null hypothesis of no treatment effect $\tau = 0$. 
         \end{algorithmic}
\end{algorithm}

\begin{ass}[Identification condition for Placebo tests] \label{ass:placebo}  Suppose that 

\begin{itemize} 
\item[(A)] For all $i, t$, 
$$
Y_{i,t}(0) \perp \Big[N_0, T_0\Big] \Big| \tilde{\Lambda}_t
$$  
\item[(B)] There exist weights $w_{v}^\star(i)$, for units $i \in \{N, \cdots, N + N_1\}$, for $N, N_1$ as in Algorithm \ref{alg:placebo}, satisfying $||w_v^\star(i)||_{\infty} = \mathcal{O}(N^{-2/3})$ and 
$$
\Big|\Big|\gamma_i - \sum_{j \neq i} w_{v, j}^\star(i) \gamma_j \Big|\Big|_2 = 0, \quad \sum_{j < N} w_{v, j}^\star(i) = 1.
$$ 
\end{itemize} 
\end{ass}  

Assumption \ref{ass:placebo} requires that the high dimensional loadings $\tilde{\Gamma}_i$ are exogenous and that we can match the endogenous loadings $\gamma_i$ of \textit{each} unit $i$ in the treated pool (where the treatment pool can be arbitrary). For given unit $j$, we construct a ``placebo" Synthetic Difference-in-Differences by taking 
\begin{equation} \label{eqn:placebo} 
\small 
\begin{aligned} 
\hat{\tau}^{pl}(j;w_v^\star, w_h^\star, v_h^\star) = \sum_{t \ge T} v_{h,t}^\star \Big\{Y_{j, t} - \sum_{s < T} w_{h,s}^\star Y_{j,s} - \sum_{i \neq j} w_{v,i}^\star(j) Y_{i,t} + \sum_{i \neq j, s< T} w_{v,i}^\star(j) w_{h,s}^\star Y_{i,s}\Big\}.
\end{aligned}
\end{equation} 
The estimator $\hat{\tau}^{pl}(j; w_v^\star, w_h^\star, v_h^\star)$ uses the same horizontal weights for each placebo units, and different vertical weights for each unit $j$. 

\begin{thm} \label{thm:place}  Suppose that Assumptions \ref{ass:outcomea}, \ref{ass:inda}, \ref{ass:iid}, \ref{ass:placebo} hold. Let $\delta_h = \mathbf{1}$. Then under the null hypothesis $H_0: \tau = 0$, for $N_1$ fixed, $T_1 \ge 1, x \in \mathbb{R}$, as $T, N \rightarrow \infty$, with $T_1/N^{1/3} = o(1)$
$$
\small 
\begin{aligned} 
& \Big| P\Big(\frac{\hat{\tau}^{pl}(n; w_v^\star, w_h^\star, v_h^\star)}{||v_h^\star||_2} \le x \Big| (\tilde{\Lambda}_t)_{t \ge 1}, T_0 = T, N_0 = N\Big) \\ &\quad \quad \quad - P\Big(\frac{\hat{\tau}^{pl}(N; w_v^\star, w_h^\star, v_h^\star)}{||v_h^\star||_2} \le x\Big| (\tilde{\Lambda}_t)_{t \ge 1}, T_0 = T, N_0 = N\Big)\Big| \rightarrow 0, 
\end{aligned}
$$ 
with $\{\hat{\tau}^{pl}(n; w_v^\star, w_h^\star, v_h^\star)/||v_h^\star||_2\}_{n \ge N}$ asymptotically independent across $n$, conditional on $(\tilde{\Lambda}_t)_{t \ge 1}, T_0 = T, N_0 = N$. 
\end{thm}

\begin{proof}[Proof of Theorem \ref{thm:place}] See Appendix \ref{sec:proof_placebo}. 
\end{proof} 

As a direct corollary of results in \cite{canay2017randomization}, we can use placebo tests as in \cite{abadie2010synthetic} to construct confidence intervals.  Here, the validity of the placebo tests relies on Assumption \ref{ass:placebo} which, together with exogeneity of $\tilde{\Gamma}_i$ imposes a symmetry restriction on the units used for the placebo test.\footnote{See also \cite{hahn2017synthetic} for a discussion on placebo tests. Results for placebo tests (studied in the different context of conformal inference) are also discussed in \cite{chernozhukov2021exact}. }

\section{Convergence rates of the weights} \label{app:convergence}

In the following proposition, we assume that the penalties for estimating the weights are 
$p_h = (\sigma_h^2 - 1)\zeta_h, p_h^\star = \sigma_h^2 \zeta_h$, where $\sigma_h^2 = \frac{1}{N - 1} \sum_{i=1}^{N-1} \Big((\tilde{\Gamma}_i + \gamma_i) \Sigma_{\tilde{\Lambda}}(\tilde{\Gamma}_i + \gamma_i)\Big) + \mathrm{Var}(\varepsilon_{i,t})$. Note that $\sigma_h^2$ does not need to be known by the researcher, since the penalty multplies by an arbitrary parameter $\zeta_h$. Properties of $\zeta_h$ affects the guarantees in the following proposition. 

\begin{prop}[Weights estimation error] \label{proof:estimation_error}  Suppose that Assumptions \ref{ass:outcomea},  \ref{ass:inda}, \ref{ass:iid}, \ref{ass:unc1a} hold. Let $\sigma_h^2 > c \times \mathrm{rank}(A)/N_0$ for a finite constant $c$ independent of $\theta^\star, N_0, T_0, N_1, T_1$.
Take any $k \ge 1$. Then conditional on $N_0 = N, T_0 = T$ and $(\tilde{\Gamma}_i)_{i \ge 1}$, 
$$
\small 
\begin{aligned} 
\zeta_h^{1/2} ||(\hat{\theta}_h - \theta_h^\star)||_2 \le c_0 s, \quad ||\hat{\beta}_h - \beta_h^\star  + A (\hat{\theta}_h - \theta_h^\star)||_2 \le c_0' \zeta_h N_0^{1/2} s
\end{aligned} 
$$   
with probability at least $1 - c' \exp(-c'' u(k,s))$, for any $s$ satisfying
$$
\small 
\begin{aligned} 
s^2 & \ge \frac{c'}{\zeta_h - c \times  \mathrm{rank}(A)/N_0}  \Big[ \frac{k^2 \mu^2 \log(T_0 + T_1)}{ N_0 } + \frac{(k^2 \mu^2 \mathrm{rank}(A)) ||\theta^\star||_2^2}{ N_0 } + \frac{||\theta^\star||_2  (k^2 \mu^2 \mathrm{rank}(A))^{1/2}}{ N_0 } \\ &\quad \quad \quad + \frac{k \mu ||A \theta_h^\star + \beta_h^\star||\log(T_0 + T_1) + \mu^2 k ||\theta^\star||_2 N_0^{1/2}  \log(T_0 + T_1)}{ N_0} \Big].
\end{aligned} 
$$ 
where $u(v,s) = \min\{v^2 \mu^2 \log(T_0 + T_1)/s^2, k^2 \mathrm{rank}(A), N_0\}$, and $c',c'', c_0, c_0' < \infty$ are finite constant independent of $\theta^\star, N_0, T_0, N_1, T_1$.  
\end{prop}

The proof follows verbatim from Theorem 1 in \cite{hirshberg2021least}, once, in our framework, we condition on $\Gamma_i = \gamma_i + \tilde{\Gamma}_i$. Proposition \ref{proof:estimation_error} is stated \textit{conditional} on $N_0 = N, T_0 = T$ and $\tilde{\Gamma}_i$. The estimation error depends on the rank of the matrix $A$ (i.e., the number of endogenous factors), but not on the rank of $\Gamma_i^\top \tilde{\Lambda}_t$.

\section{Proofs}

Let $\Gamma_i = \tilde{\Gamma}_i + \gamma_i$, $\Lambda_t = \tilde{\Lambda}_t + \lambda_t$. Recall that $(\gamma_i)_{i \ge 1}, (\lambda_t)_{t \ge 1}$ are deterministic. 

\subsection{Proof of Proposition \ref{lem:3a}} \label{proof:lem:3a}
Assumption \ref{ass:outcomea} guarantees that the potential outcome under control follows a factor model. Assumption \ref{ass:unc1a} guarantees that we can match $\sum_{t \ge T} v_{h,t} \lambda_t = \sum_{s < T} w_{h,s} \lambda_s$ and that $\mathbb{E}[\tilde{\Lambda}_t | T_0 = T] = \mathbb{E}[\tilde{\Lambda}_t]$. Assumption \ref{ass:inda} guarantees that $\tilde{\Lambda}_t$ is independent of $N_0$. The restriction that the weights sum to one in Assumption \ref{ass:unc1a} guarantees that $\sum_{t \ge T} v_{h,t} \iota_{0, i} = \sum_{s < T} w_{h, s} \iota_{0, i}$ (the weighted combination of unit fixed effects) is the same on the right and left-hand side of Equation \eqref{eqn:ajh}. Proposition \ref{lem:3a} holds for $\beta_{0, h}(w_h, v_h) = \sum_{t \ge T} v_{h,t} \iota_{1, t} - \sum_{s < T} w_{h,s} \iota_{1, s}$. 

\subsection{Proof of Proposition \ref{cor:2}} \label{proof:dr}

\paragraph{Error for unit $(N,T$)} We first focus on unit $(N,T)$ whereas the reasoning applies to all other units $i \ge N, t \ge T$. Define 
$$
\small 
\begin{aligned} 
\bar{\Lambda}_T(w_v) = \sum_{s < T} w_{v, j} \mathbb{E}\Big[\Lambda_s \Big| T_0 = T, N_0 = N, \Lambda_T\Big], \bar{\Gamma}_N(w_h)  = \sum_{j < N} w_{v, j} \mathbb{E}\Big[\Gamma_j \Big| \Gamma_N, T_0 = T, N_0 = N\Big] . 
\end{aligned} 
$$ 
Because of Assumptions \ref{ass:outcomea} and \ref{ass:inda}, we can write 
$$
\small 
\begin{aligned} 
\sum_{s < T} w_{h,s} \mathbb{E}\Big[Y_{N,s} \Big| T_0 = T, N_0 = N, \Gamma_N, \Lambda_T\Big]  & = \Gamma_N \sum_{s < T} w_{h,s} \mathbb{E}\Big[\Lambda_s | \Lambda_T, T_0 = T, N_0 = N\Big] + \sum_{s <T} w_{h,s} \iota_{1, s} + \iota_{0,N} \\
& = \Gamma_N \bar{\Lambda}_T(w_h) + \sum_{s <T} w_{h,s} \iota_{1, s} + \iota_{0,N} \\ 
 \sum_{j < N} w_{v, j} \mathbb{E}\Big[Y_{j,T} \Big| T_0 = T, N_0 = N, \Gamma_N, \Lambda_T\Big] & = \Lambda_T \sum_{j < N} w_{v, j} \mathbb{E}\Big[\Gamma_j \Big| \Gamma_N, T_0 = T, N_0 = N\Big] + \sum_{j < N} w_{v, j} \iota_{0, j} + \iota_{1,T} \\
&  = \Lambda_T \bar{\Gamma}_N(w_v) + \sum_{j < N} w_{v, j} \iota_{0, j} + \iota_{1,T} . 
\end{aligned} 
$$  
In addition,
$$
\small 
\begin{aligned}
& \sum_{s<T} \sum_{j < N} w_s w_{v, j} \mathbb{E}\Big[Y_{j,s} \Big| T_0 = T, N_0 = N, \Gamma_N, \Lambda_T\Big]   \\ &= \sum_{j < N} w_{v, j} \mathbb{E}\Big[\Gamma_j \Big| \Gamma_N, T_0 = T, N_0 = N\Big]  \sum_{s < T} w_{h,s} \mathbb{E}\Big[\Lambda_s \Big| \Lambda_T, T_0 = T, N_0 = N\Big] \quad (\because \text{ Assumption } \ref{ass:inda}) \\ & + \sum_{s < T} w_{h, s} \iota_{1, s} + \sum_{j < N} w_{v,j} \iota_{0, j} = \bar{\Gamma}_N(w_v) \bar{\Lambda}_T(w_h) + \sum_{s < T} w_{h, s} \iota_{1, s} + \sum_{j < N} w_{v, j} \iota_{0, j}.  
\end{aligned} 
$$ 
It follows that 
$$
\small 
\begin{aligned} 
 & \sum_{s < T} w_{h, s} \mathbb{E}\Big[Y_{N,s} \Big| T_0 = T, N_0 = N, \Gamma_N, \Lambda_T\Big] + \sum_{j < N} w_{v, j} \mathbb{E}\Big[Y_{j,T} \Big| T_0 = T, N_0 = N, \Gamma_N, \Lambda_T\Big] \\ &-  \sum_{s<T} \sum_{j < N} w_s w_{v, j} \mathbb{E}\Big[Y_{j,s} \Big| T_0 = T, N_0 = N, \Gamma_N, \Lambda_T\Big] - \Big(\Lambda_T \Gamma_N + \iota_{0,N} + \iota_{1, T}\Big) \\ 
&= \Gamma_N \bar{\Lambda}_T(w_h) + \bar{\Gamma}_N(w_v) \Lambda_T - \bar{\Gamma}_N(w_v)\bar{\Lambda}_T(w_h) - \Lambda_T \Gamma_N \\ 
&= (\Gamma_N - \bar{\Gamma}_N(w_v))(\bar{\Lambda}_T(w_h) - \Lambda_T) .  
\end{aligned} 
$$ 
The equation above directly extends to any unit $(i, t), i \ge N, t \ge T$. The final result follows by first taking expectations of $(\Gamma_i - \bar{\Gamma}_i(w_v))(\bar{\Lambda}_i(w_h) - \Lambda_i), i \ge N, t \ge T$ over $\tilde{\Gamma}_i, \tilde{\Lambda}_t$ and using Assumption \ref{ass:inda}; and then averaging over treated units and treatment periods, with corresponding weights $v_v, v_h$.

\subsection{Proof of Theorem \ref{thm:horizontal_asym}} \label{proof:thm_hor1}

We break the proof into multiple steps, where we define $\Gamma_i = \tilde{\Gamma}_i + \gamma_i$, $\sigma_{\varepsilon}^2 = \mathbb{E}[\varepsilon_{i,t}^2]$. Recall that $\mathbb{E}\Big[\hat{\tau}^h(w_h^\star, v_h^\star, \beta_h^\star, q_h) - \tau \Big| T_0 = T, N_0 = N, (\Gamma_i)_{i\ge 1}\Big] = 0$ by Proposition \ref{lem:expectations}.

\paragraph{Order of convergence for post-treatment period}  First, we claim that 
\begin{equation} \label{eqn:pre}
\sum_{s < T} w_{n,s}^\star \Big\{Y_{n,s} - \mathbb{E}\left[Y_{n,s} | N_0 = N, T_0 = T, (\Gamma_i)_{i\ge 1}\right]\Big\} = O_p(||w_h^\star||_2) = O_p(T_0^{-1/6}). 
\end{equation}
To show this claim we use Assumption \ref{ass:outcomea} and Assumption \ref{ass:horizontal_inference} (A). First, we write 
$$
\sum_{s < T} w_{n,s}^\star \Big\{Y_{n,s} - \mathbb{E}\left[Y_{n,s} | N_0 = N, T_0 = T, (\Gamma_i)_{i\ge 1}\right] \Big\}= \Gamma_n^\top \sum_{s < T} w_{n,s}^\star \tilde{\Lambda}_s + \sum_{s < T} w_{n,s}^\star \varepsilon_{n,s}. 
$$ 
By Assumptions \ref{ass:inda},  \ref{ass:iid} 
$$
\begin{aligned} 
\mathbb{V}\left( \Gamma_n^\top \sum_{s < T} w_{n,s}^\star \tilde{\Lambda}_s + \sum_{s < T} w_{n,s}^\star \varepsilon_{n,s} \Big| N_0 = N, T_0 = T, (\Gamma_i)_{i\ge 1} \right) & &\le  ||\Gamma_n||_2^2 \mathbb{E}[||\tilde{\Lambda}_t||_2^2] ||w_h^\star||_2^2 + ||w_h^\star||_2^2 \sigma_\varepsilon^2.
\end{aligned}
$$ 
Therefore since $||w_h^\star||_{\infty} \le T_0^{-2/3}$, and from Assumption \ref{ass:iid},
$$
\sum_{s < T} w_{n,s}^\star \Big\{ Y_{n,s} - \mathbb{E}\left[Y_{n,s} | N_0 = N, T_0 = T, (\Gamma_i)_{i\ge 1}\right] \Big\} = O_p(||w_h^\star||_2) = O_p(T_0^{-1/6}). 
$$ 
\paragraph{CLT for fixed unit $n$} We can write 
$$
\begin{aligned}
& \sum_{t \ge T} v_{h,t}^\star Y_{n,t}  - \sum_{t \ge T} v_{h,t}^\star \mathbb{E}[Y_{n,t} | \Gamma_n, N_0 = N, T_0 = T]  = \Gamma_n^\top \sum_{t \ge T} v_{h,t}^\star \tilde{\Lambda}_t + \sum_{t \ge T} v_{h,t}^\star \varepsilon_{n,t} .  
\end{aligned} 
$$ 
We can use the Lyaponuv's central limit theorem here. In particular, we have under Assumption \ref{ass:inda}
$$
\small 
\begin{aligned} 
& \mathbb{E}\Big[\Big(\Gamma_n^\top \sum_{t \ge T} v_{h,t}^\star \tilde{\Lambda}_t + \sum_{t \ge T} v_{h,t}^\star \varepsilon_{n,t} \Big)^3 \Big| (\Gamma_i)_{i \ge 1}, N_0 = N, T_0 = T\Big] \\ & = \mathbb{E}\Big[\Big(\Gamma_n^\top \sum_{t \ge T} v_{h,t}^\star \tilde{\Lambda}_t\Big)^3 + \Big(\sum_{t \ge T} v_{h,t}^\star \varepsilon_{n,t}\Big)^3 \Big| (\Gamma_i)_{i \ge 1}, N_0 = N, T_0 = T\Big]  \\ & = \sum_{t \ge T} v_{h,t}^{\star, 3}  \mathbb{E}[(\Gamma_n^\top \tilde{\Lambda}_t)^3 | \Gamma_n] + \sum_{t \ge T} v_{h,t}^{\star, 3} \mathbb{E}[\varepsilon_{i,t}^3] \quad (\because \text{ Assumption } \ref{ass:unc1a}) \\ & =  
\mathcal{O}\Big(||v_h^\star||_{\infty} ||v_{h}^\star||_2^2 \Big) \quad (\because \text{ Assumption } \ref{ass:iid}) 
\end{aligned} 
$$ 
Also, by Assumption \ref{ass:inda}
$$
\mathrm{Var}\left(\Gamma_n \sum_{t \ge T} v_{h,t} \tilde{\Lambda}_t + \sum_{t \ge T} v_{h,t} \varepsilon_{n,t} \Big| (\Gamma_i)_{i \ge 1}, T_0 = T, N_0 = N\right) \ge ||v_h^\star||_2^2 \sigma_{\varepsilon}^2.
$$ 
It follows 
$$
\small 
\begin{aligned} 
\frac{\mathbb{E}\Big[\Big(\Gamma_n \sum_{t \ge T} v_{h,t} \tilde{\Lambda}_t + \sum_{t \ge T} v_{h,t} \varepsilon_{n,t}\Big)^3 \Big| (\Gamma_i)_{i \ge 1}, N_0 = N, T_0 = T\Big]}{\mathrm{Var}\left(\Gamma_n \sum_{t \ge T} v_{h,t} \tilde{\Lambda}_t + \sum_{t \ge T} v_{h,t} \varepsilon_{n,t} \Big| (\Gamma_i)_{i \ge 1}, T_0 = T, N_0 = N\right)^{3/2}} = \mathcal{O}(T_1^{-2/3} ||v_h^\star||_2^{-1}). 
\end{aligned} 
$$ 
Because $||v_h^\star||_2 \ge 1/\sqrt{T}_1$ for any $1^\top v_h^\star = 1$, by Lyapunov's central limit theorem for any $n$, 
$$
\frac{1}{\sqrt{||v_h^\star||_2^2 \Gamma_n^\top \Sigma_{\tilde{\Lambda}} \Gamma_n + \sigma_\varepsilon^2 ||v_h^\star||_2^2}} \Big(\sum_{t \ge T} v_{h,t}^\star Y_{n,t}  - \sum_{t \ge T} v_{h,t}^\star \mathbb{E}[Y_{n,t} | \Gamma_n, N_0 = N, T_0 = T] \Big)  \rightarrow \mathcal{N}(0, 1), 
$$ 
as $T_1 \rightarrow \infty$.

\paragraph{Central limit theorem after summing over $N_1$ treated units} 
Consider now summing $N_1$ treated units for finite $N_1$, taking 
$$
\sum_{n \ge N} q_{h, n} \Big\{\sum_{t \ge T} v_{h,t}^\star Y_{n,t}   - \sum_{t \ge T} v_{h,t}^\star \mathbb{E}[Y_{n,t} | \Gamma_n, N_0 = N, T_0 = T]\Big\}. 
$$
It follows that as $T_1 \rightarrow \infty$, for given $N_1$, we can show using the same technique as above that where $\Gamma_n$ is replaced by $\sum_{n \ge N} q_{h,n} \Gamma_n$ and $\varepsilon_{n,t}$ by $\sum_{n \ge N} q_{h,n} \varepsilon_{n,t}$, 
\begin{equation} \label{eqn:variance_to_bound} 
\small 
\begin{aligned} 
\frac{\sum_{n \ge N} q_{h, n} \Big\{\sum_{t \ge T} v_{h,t}^\star Y_{n,t}   - \sum_{t \ge T} v_{h,t}^\star \mathbb{E}[Y_{n,t} | \Gamma_n, N_0 = N, T_0 = T]\Big\}}{\sqrt{||v_h^\star||_2^2 \Big(\sum_{n \ge N} q_{h,n} \Gamma_n\Big)^\top \Sigma_{\tilde{\Lambda}} \Big(\sum_{n \ge N} q_{h,n} \Gamma_n\Big) + ||v_h^\star||_2^2 ||q_h||_2^2 \sigma_{\varepsilon}^2}}  \rightarrow_d \mathcal{N}\Big(0, 1\Big).
\end{aligned} 
\end{equation}

\paragraph{Lower bound on the variance component} Finally, we show that the component in Equation \eqref{eqn:pre} converges at a faster rate than the component in Equation \eqref{eqn:variance_to_bound}. In particular,
$$
\small 
\begin{aligned} 
||v_h^\star||_2^2 \Big(\sum_{n \ge N} q_{h,n} \Gamma_n\Big)^\top \Sigma_{\tilde{\Lambda}} \Big(\sum_{n \ge N} q_{h,n} \Gamma_n\Big) + ||v_h^\star||_2^2 ||q_h||_2^2 \sigma_{\varepsilon}^2 \ge \frac{\sigma_{\varepsilon}^2}{N_1 T_1}
\end{aligned} 
$$ 
since the weights sum to one, and therefore $||v_h^\star||_2^2 \ge 1/T_1, ||q_h||_2^2 \ge 1/N_1$. As a result, because $N_1 T_1/T_0^{1/3} = o(1)$ by Assumption \ref{ass:horizontal_inference}, and $T_0^{-1/3}$ is the order of convergence of the variance of the component in Equation \eqref{eqn:pre}, the final result directly follows.

\subsection{Proof of Theorem \ref{thm:dr}} \label{proof:long}

We study asymptotic normality if either Assumption \ref{ass:unc1a} or Assumption \ref{ass:unc2a} hold. Suppose Assumption \ref{ass:unc1a} holds first, and therefore we condition on $(\Gamma_i)_{i \ge 1}$, with $\Gamma_i = \tilde{\Gamma}_i + \gamma_i$.

\paragraph{Preliminaries when Assumption \ref{ass:unc1a} (and we can condition on $\Gamma_i$)} Under Assumption \ref{ass:unc1a}, we have
$$
\small 
\begin{aligned} 
\mathbb{E}\Big[\hat{\tau}(w_h^\star, w_v^\star, v_h^\star, v_v^\star) - \tau \Big| N_0 = N, T_0 = T, (\Gamma_i)_{i \ge 1}\Big] = 0. 
\end{aligned}
$$  
We now decompose $\hat{\tau}(w_h^\star, w_v^\star, v_h^\star, v_v^\star)$ in different components and study their convergence rates. 

\paragraph{Asymptotically negligible components of $\hat{\tau}(w_h^\star, w_v^\star, v_h^\star, v_v^\star)$}
First we study two components of $\hat{\tau}(w_h^\star, w_v^\star, v_h^\star, v_v^\star)$ and show that their convergence rate is of order $\mathcal{O}_p(T_0^{-1/6})$. The first component is
$$
\small 
\begin{aligned} 
& \sum_{i < N} \sum_{s < T} w_{v, i}^\star w_{h, s}^\star \Big\{Y_{i,s} - \mathbb{E}\Big[Y_{i,s} | N_0 = N, T_0 = T, (\Gamma_i)_{i \ge 1}\Big]\Big\}
\\&= \sum_{i < N} w_{v, i}^\star \Gamma_i \sum_{s < T}  w_{h, s}^\star \tilde{\Lambda}_s +  \sum_{i < N} \sum_{s < T} w_{v, i}^\star w_{h, s}^\star \varepsilon_{i,s}. 
\end{aligned}
$$ 
We have from Assumption \ref{ass:iid}
$$
\small 
\begin{aligned} 
& \mathrm{Var}\left(\sum_{i < N} w_{v, i}^\star \Gamma_i \sum_{s < T}  w_{h, s}^\star \tilde{\Lambda}_s +  \sum_{i < N} \sum_{s < T} w_{v, i}^\star w_{h, s}^\star \varepsilon_{i,s} \Big| T_0 = T, N_0 = N, (\Gamma_i)_{i \ge 1}\right) \\ &\le \sigma_{\varepsilon}^2 ||w_v^\star||_2^2 ||w_h^\star||_2^2 + ||w_h^\star||_2^2 \mathcal{O}(1) = \mathcal{O}(||w_h^\star||_2^2). 
\end{aligned} 
$$ 
Because $||w_h^\star||_2^2 \le T^{-1/3}$, it follows that  
$$
\sum_{i < N} \sum_{s < T} w_{v, i}^\star w_{h, s}^\star \Big\{Y_{i,s} - \mathbb{E}\Big[Y_{i,s} | N_0 = N, T_0 = T, (\Gamma_i)_{i \ge 1}\Big]\Big\} = \mathcal{O}_p(T^{-1/6}).
$$ 
The second component is 
$$
\sum_{s < T} w_{h, s}^\star \Big\{ Y_{n,s} - \mathbb{E}\Big[Y_{n,s} \Big|T_0 =T, N_0 = N, (\Gamma_i)_{i \ge 1}\Big]\Big\}. 
$$ 
We write 
$$
\small 
\begin{aligned} 
\mathrm{Var}\left(\sum_{s < T} w_{h, s}^\star Y_{n,s} \Big| (\Gamma_i)_{i \ge 1}, T = T_0 , N = N_0\right) = \mathcal{O}(||w_h^\star||_2^2) = \mathcal{O}(T^{-1/3}). 
\end{aligned} 
$$ 
Therefore, it follows that 
$$
\small 
\begin{aligned} 
\sum_{s < T} w_{h, s}^\star \Big\{Y_{n,s} - \mathbb{E}\Big[Y_{n,s} | (\Gamma_i)_{i \ge 1}, T_0 = T, N_0 = N\Big]\Big\} = \mathcal{O}_p(T^{-1/6}). 
\end{aligned} 
$$ 

\paragraph{Asymptotic normality for fixed $n$} 
Consider now, for given $n$,
\begin{equation} \label{eqn:11} 
\small 
\begin{aligned} 
& \sum_{t \ge T} v_{v,t}^\star \Big\{ Y_{n,t} - \mathbb{E}[Y_{n,t} | (\Gamma_i)_{i \ge 1}, T_0 = T, N_0 = N] - \sum_{i < N} w_{v,i}^\star (Y_{i, t} - \mathbb{E}[Y_{i,s} | (\Gamma_i)_{i \ge 1}, N_0 = N, T_0 = T])\Big\} \\ &= \Big(\Gamma_n - \sum_{i < N} w_{v,i}^\star \Gamma_i\Big) \sum_{t \ge T} v_{v,t}^\star \tilde{\Lambda}_t + \sum_{t \ge T} v_{v,t}^\star \varepsilon_{n,t}. 
\end{aligned} 
\end{equation}
We can write 
\begin{equation} \label{eqn:variance} 
\small 
\begin{aligned} 
\mathrm{Var}\left(\Big(\Gamma_N - \sum_{i < N} w_{v,i}^\star \Gamma_i\Big) \sum_{t \ge T} v_{v,t}^\star \tilde{\Lambda}_t + \sum_{t \ge T} v_{v,t}^\star \varepsilon_{n,t} \Big| T_0 = T, N_0 = N, (\Gamma_i)_{i \ge 1}\right) \ge ||v_v^\star||_2^2 \sigma_\varepsilon^2 \ge \sigma_\varepsilon^2/T_1,
\end{aligned} 
\end{equation}   
since $\sigma_{\varepsilon}^2 > 0$ and $||v_v^\star||_2^2 \ge 1/T_1$ (since $1^\top v_v = 1$).

We can follow verbatim the proof of Theorem \ref{thm:horizontal_asym} (paragraph ``CLT for fixed $n$), with $(\Gamma_n - \sum_{i < N} w_{v, i}^\star \Gamma_i)$ in lieu of $\Gamma_n$ in the proof of Theorem \ref{thm:horizontal_asym}, and obtain 
\begin{equation} \label{eqn:lyapunov}
\small 
\begin{aligned}
\frac{\mathbb{E}\Big[\Big\{\Big(\Gamma_n - \sum_{i < N} w_{v,i}^\star \Gamma_i\Big) \sum_{t \ge T} v_{v,t}^\star \tilde{\Lambda}_t + \sum_{t \ge T} v_{v,t}^\star \varepsilon_{n,t}\Big\}^3 \Big| (\Gamma_i)_{i \ge 1}, N_0 = N, T_0 = T\Big]}{\mathrm{Var}\left(\Big(\Gamma_n - \sum_{i < N} w_{v,i}^\star \Gamma_i\Big) \sum_{t \ge T} v_{v,t}^\star \tilde{\Lambda}_t + \sum_{t \ge T} v_{v,t}^\star \varepsilon_{n,t} \Big| T_0 = T, N_0 = N, (\Gamma_i)_{i \ge 1}\right)^{3/2}}   = o(1). 
\end{aligned} 
\end{equation}

\paragraph{Collecting the terms} We can write, as $T_1 \rightarrow \infty$, for any $N_1$, 
\begin{equation} \label{eqn:above} 
\small 
\begin{aligned}  
\hat{\tau}(w_h^\star, w_v^\star, v_h^\star, v_v^\star) - \tau = \sum_{n \ge N} v_{h, n}\Big\{\underbrace{\Big(\Gamma_n - \sum_{i < N} w_{v,i}^\star \Gamma_i\Big) \sum_{t \ge T} v_{v,t}^\star \tilde{\Lambda}_t + \sum_{t \ge T} v_{v,t}^\star \varepsilon_{n,t}}_{(A)} + \mathcal{O}_p(T_0^{-1/6})\Big\}.
\end{aligned} 
\end{equation}  
Therefore, because $T_1 N_1/T_0^{-1/3} = o(1)$, it follows from Equation \eqref{eqn:variance} that $(A)$ has a slower convergence rate than $T_0^{-1/6}$ and the second component in the right-hand side of Equation \eqref{eqn:above} is asymptotically negligible relative to the first one. 

\paragraph{Asymptotic normality and conclusions} Here, $(A)$ in Equation \eqref{eqn:above} is asymptotically normal by Equation \eqref{eqn:lyapunov} and Lyapunov's central limit theorem. 
Because $\varepsilon_{i,t}$ are independent across units and time, of $\Lambda_t$, and because $v_h^\top 1 = 1$, the central limit theorem for Equation \eqref{eqn:above} directly applies for the variance as defined in Theorem \ref{thm:dr}.

\paragraph{Conclusions}
The case under Assumption \ref{ass:unc2a} follows similarly after inverting the role of the loadings and factors.

\subsection{Proof of Theorem \ref{thm:place}} \label{sec:proof_placebo} Under Assumption \ref{ass:placebo}, following verbatim the proof of Theorem \ref{thm:dr} while inverting the loadings and factors, we can write under the null that $\tau = 0$ (letting $\Gamma_i = \gamma_i + \tilde{\Gamma}_i$)
$$
\hat{\tau}^{pl}(j; w_v^\star, w_h^\star, v_h^\star) = \sum_{t \ge T} v_{h,t}^\star \Big\{\Gamma_j \tilde{\Lambda}_t + \varepsilon_{j,t} - \Gamma_j \sum_{s < T} w_{h,s}^\star \tilde{\Lambda}_s\Big\} + \mathcal{O}_p(N_0^{-1/6}). 
$$ 
We can then follow verbatim the proof of Theorem \ref{thm:dr} (for fixed unit $n$), after inverting the factors with the loadings to claim asymptotic normality. Exchangeability of each estimator $\hat{\tau}^{pl}(j; w_v^\star, w_h^\star, v_h^\star)$ directly follows from Assumption \ref{ass:iid}.

\subsection{Proof of Corollary \ref{lem:aa}} \label{proof:lemma_convergence} We can write 
$$
\small 
\begin{aligned} 
\frac{\hat{\tau}^h(\hat{w}_h, \hat{v}_h, \hat{\beta}_h) - \tau}{\bar{\mathbb{V}}_h^{1/2}(N, T, q_h)}  =  \underbrace{\frac{\hat{\tau}^h(w_h^\star, v_h^\star, \beta_h^\star) - \tau}{\bar{\mathbb{V}}_h^{1/2}(N, T, q_h)}}_{(A)}  + \underbrace{\frac{\hat{\tau}^h(\hat{w}_h, \hat{v}_h, \hat{\beta}_h) - \hat{\tau}^h(w_h^\star, v_h^\star, \beta_h^\star)}{\bar{\mathbb{V}}_h^{1/2}(N, T, q_h)}}_{(B)}.
\end{aligned} 
$$ 
By Theorem \ref{thm:horizontal_asym}, we have $(A) \rightarrow_d \mathcal{N}(0,1)$. We have $\bar{V}_h(N, T, q_h) \ge ||q_h||^2 ||v_h^\star||^2 \sigma_\varepsilon^2 \ge T_1 N_1 \sigma_\varepsilon^2$. Therefore $(B) = o_p(1)$ under Equation \eqref{eqn:error}. The result follows from Slutsky theorem.

\end{document}